\documentclass[10pt,journal,twoside,compsoc]{IEEEtran}

\usepackage{soul}
\usepackage{amssymb}
\usepackage{amsfonts,amsmath}
\usepackage{mathrsfs}
\usepackage{bbm}
\usepackage{bm}
\usepackage{arydshln}
\usepackage[ruled, linesnumbered]{algorithm2e}
\usepackage[noend]{algpseudocode}
\usepackage{amsthm}
\usepackage{nccmath}
\usepackage{booktabs,siunitx}
\usepackage{adjustbox}
\usepackage{rotating}
\usepackage{threeparttable}
\usepackage{caption}
\usepackage{multirow}
\usepackage{subcaption}
\usepackage{comment}
\usepackage{graphicx}
\usepackage{lipsum,mwe}
\usepackage{enumerate}
\usepackage{enumitem}
\usepackage{nicefrac}
\usepackage{xfrac}
\usepackage{listings,lstautogobble}
\usepackage{tikz}
\usetikzlibrary{fit,shapes,arrows,patterns}
\usepackage{tcolorbox}
\usepackage{pdflscape}
\usepackage{afterpage}
\usepackage{pifont}
\usepackage{hyperref}
\usepackage{aligned-overset}
\usepackage{color, colortbl}
\PassOptionsToPackage{usenames,dvipsnames,svgnames,table}{xcolor}

\DeclareMathOperator{\sign}{sign}

\DeclareMathOperator{\minimum}{minimum}
\DeclareMathOperator*{\argmax}{arg\,max}

\DeclareMathSymbol{\mh}{\mathord}{operators}{`\-}

\DeclareMathOperator{\round}{round}

\DeclareMathOperator{\proj}{Proj}
\DeclareMathOperator{\ce}{\mathcal{F}}
\DeclareMathOperator{\de}{\mathcal{G}}

\newcommand{\cc}{\cellcolor{lightgray!30}}

\makeatletter
\newcommand{\removelatexerror}{\let\@latex@error\@gobble}
\newtheorem{theorem}{Theorem}
\newtheorem*{theorem*}{Theorem}

\newtheorem{lemma}{Lemma}

\newtheorem{assumption}{Assumption}
\newtheorem{proposition}{Proposition}
\newtheorem*{proposition*}{Proposition}

\newcommand{\ignore}[1]{}

\newcommand{\specialcell}[2][c]{\begin{tabular}[#1]{@{}c@{}}#2\end{tabular}}
\newcommand{\specialcellleft}[2][l]{\begin{tabular}[#1]{@{}l@{}}#2\end{tabular}}

\newcommand{\vthead}[1]{\makebox[1em][l]{\rotatebox{90}{#1}}}

\newcommand*\circled[1]{
	\begin{tikzpicture}[scale=1., baseline=(char.base)] 
		\node[shape=circle,fill,inner sep=1.5pt] (char) {\textcolor{white}{#1}};
	\end{tikzpicture}
}

\newcommand*\whitecircled[1]{\tikz[baseline=(char.base)]{
		\node[shape=circle,draw,inner sep=0.8pt] (char) {#1};}}

\definecolor{dkgreen}{rgb}{0,0.6,0}
\definecolor{mauve}{rgb}{0.58,0,0.82}

\ifCLASSOPTIONcompsoc
  \usepackage[nocompress]{cite}
\else
  \usepackage{cite}
\fi

\begin{document}
\title{PAD: Towards Principled Adversarial Malware Detection Against Evasion Attacks} 

\author{Deqiang Li,
        Shicheng Cui,
        Yun Li,
        Jia Xu,
        Fu Xiao
        and Shouhuai Xu % <-this % stops a space
\IEEEcompsocitemizethanks{\IEEEcompsocthanksitem D. Li is with the School of Computer Science, Nanjing University of Posts and Telecommunications, Nanjing, 210023, China
\IEEEcompsocthanksitem S. Cui is with the School of Computer Engineering, Nanjing Institute of Technology, Nanjing, 211167, China 
\IEEEcompsocthanksitem Y. Li, J. Xu, and F. Xiao are with the School of Computer Science, Nanjing University of Posts and Telecommunications, Nanjing, 210023, China 
\IEEEcompsocthanksitem S. Xu is with Department of Computer Science, University of Colorado Colorado Springs, 1420 Austin Bluffs Pkwy, Colorado Springs, Colorado, 80918 USA. \protect \\Email: sxu@uccs.edu
}
\thanks{More information can be found at \url{http://ieeexplore.ieee.org}}
}

\markboth{IEEE TRANSACTIONS ON DEPENDABLE AND SECURE COMPUTING, ~Vol.~XX, No.~X, August~2023}{PAD: Principled Adversarial Malware Detection Against Evasion Attacks, D. Li, S. Cui, Y. Li, J. Xu, F. Xiao and S. Xu.}

\IEEEpubid{1545-5971~\copyright~2023 IEEE. Personal use of this material is permitted, but republication/redistribution requires IEEE permission.}

\IEEEtitleabstractindextext{%
\begin{abstract}
Machine Learning (ML) techniques can facilitate the automation of \underline{mal}icious soft\underline{ware} (malware for short) detection, but suffer from evasion attacks. Many studies counter such attacks in heuristic manners, lacking theoretical guarantees and defense effectiveness. In this paper, we propose a new adversarial training framework, termed \underline{P}rincipled \underline{A}dversarial Malware \underline{D}etection (PAD), which offers convergence guarantees for robust optimization methods. PAD lays on a learnable convex measurement that quantifies distribution-wise discrete perturbations to protect malware detectors from adversaries, 
whereby for smooth detectors, adversarial training can be performed with theoretical treatments. To promote defense effectiveness, we propose a new mixture of attacks to instantiate PAD to enhance deep neural network-based measurements and malware detectors. Experimental results on two Android malware datasets demonstrate: (i) the proposed method significantly outperforms the state-of-the-art defenses; (ii) it can harden ML-based malware detection against 27 evasion attacks with detection accuracies greater than 83.45\%, at the price of suffering an accuracy decrease smaller than 2.16\% in the absence of attacks;
(iii) it matches or outperforms many anti-malware scanners in VirusTotal against realistic adversarial malware.
\end{abstract}

}

\maketitle

\IEEEdisplaynontitleabstractindextext

\IEEEraisesectionheading{\section{Introduction}\label{sec:introduction}}

\IEEEPARstart{I}{nternet} is widely used for connecting various modern devices, which facilitates the communications of our daily life but can spread cyber attacks at the same time. For example, Kaspersky \cite{kasparsky:Online} reported detecting 33,412,568 malware samples in the year of 2020, 64,559,357 in 2021, and 109,183,489 in 2022. The scale of this threat motivates the use of Machine Learning (ML) techniques, including Deep Learning (DL), to automate malware detection. Promisingly, empirical evidence demonstrates the advanced performance of ML-based detection (see, e.g., \cite{raff2017malware,DBLP:journals/csur/YeLAI17,10.1145/3372297.3417291,hou2017hindroid,onwuzurike2019mamadroid}).

Unfortunately, ML-based malware detectors are vulnerable to {\em adversarial examples}. These examples are a type of malware variants and are often generated by modifying non-functional instructions in the existing executable programs (rather than writing them from scratch) \cite{8782574,Chen:2017:SES,li2020adversarial,pierazzi2019intriguing,DBLP:conf/ccs/ZhaoZZZZLYYL21,DBLP:conf/asiaccs/SongLAGKY22}. Adversarial examples can be equipped with {\em poisoning attacks} \cite{DBLP:journals/compsec/ChenXFHXZL18,suciu2018does}, {\em evasion attacks} \cite{DBLP:journals/tifs/DemetrioBLRA21,DBLP:journals/tissec/DemetrioCBLAR21,DBLP:conf/asiaccs/SongLAGKY22}, or both \cite{demontis2018intriguing}. In this paper, we focus on evasion attacks, which aim to mislead a malware detection model in the test phase. To combat evasive attacks, pioneers proposed several approaches, such as input transformation \cite{YeFOSINT-SI-2018}, weight regularization \cite{7917369}, and classifier randomization \cite{DBLP:journals/corr/abs-2005-11671}, most of which, however, have been broken by sophisticated attacks (e.g.,  \cite{carlini2017adversarial,DBLP:journals/corr/abs-1802-00420,pierazzi2019intriguing,DBLP:journals/corr/abs-2106-15023}). Nevertheless, recent studies empirically demonstrate that {\em adversarial training} can harden ML models to certain extent \cite{al2018adversarial,DBLP:journals/tnse/LiLYX21}, which endows a model with robustness by learning from adversarial examples, akin to ``vaccines". 

Figure \ref{fig:limit} illustrates the schema of adversarial training. Owing to the efficiency of mapping representation perturbations back to the problem space, researchers conduct adversarial training in the feature space \cite{rndic_laskov,al2018adversarial,DBLP:journals/tnse/LiLYX21,pierazzi2019intriguing,DBLP:journals/tifs/DemetrioBLRA21}. However, ``side-effect'' features \cite{pierazzi2019intriguing} cause inaccuracy when conducting the inverse representation-mapping, leading to the robustness gap that the attained robustness cannot propagate to the problem space. In the feature space, adversarial training typically involves inner maximization (searching perturbations) and outer minimization (optimizing model parameters). Both are handled with heuristic methods, lacking  theoretical guarantees \cite{al2018adversarial,DBLP:journals/tnse/LiLYX21}. This leads to the limitation of disallowing a rigorous analysis on the types of attacks that can be thwarted by the resultant model, especially in the context of discrete domains (e.g., malware detection). The fundamental concern is the optimization convergence: the inner maximization shall converge to a stationary point, and the resultant perturbation approaches the optimal one; the outer minimization has gradients of loss w.r.t. parameters proceeding toward zero regarding certain metrics (e.g., $\ell_2$ norm) in gradient-based optimization. Intuitively, as long as convergence requirements are met, the defense model can mitigate other attacks less effectively than
the one that is used for adversarial training.

\begin{figure}[!htbp]
	\centering
\includegraphics[width=0.486\textwidth]{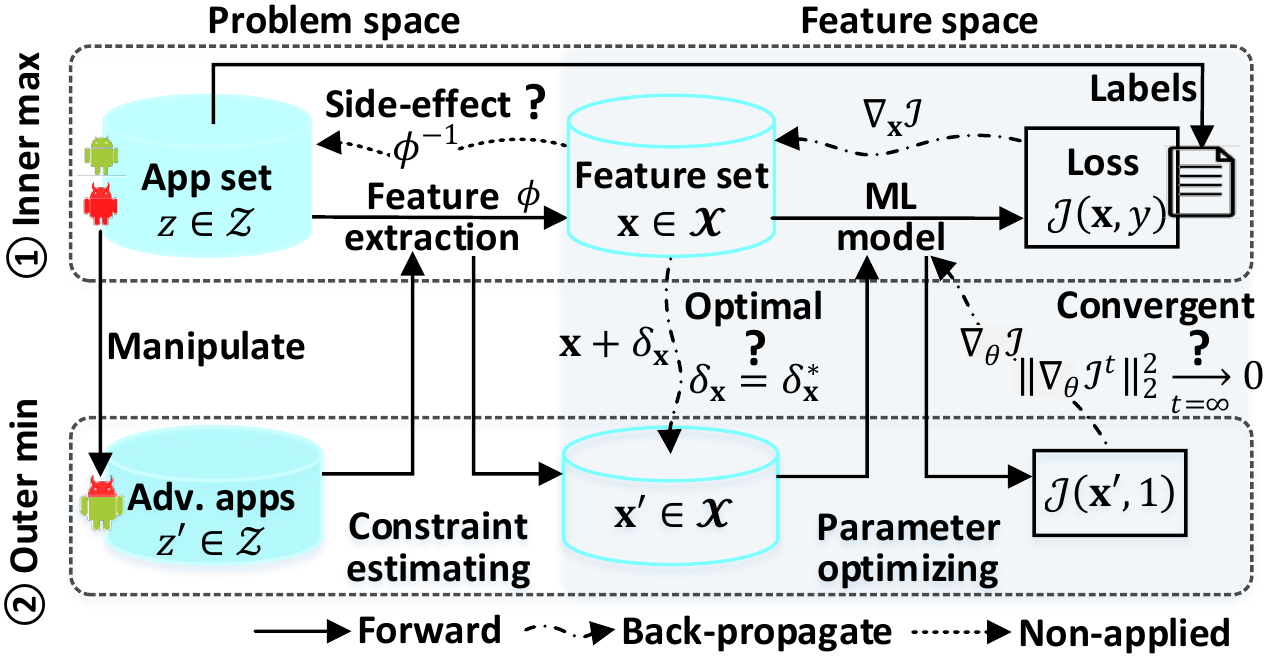}
	\caption{Schema of feature space adversarial training and its three limitations related that: (i) the attained robustness back-propagates to the problem space (upper left); (ii) the inner maximization searches perturbations optimally (middle); (iii) the outer minimization optimizes model parameters convergently (right).}
	\label{fig:limit}
\end{figure}

Existing methods cope with the limitations mentioned above with new assumptions \cite{DBLP:conf/codaspy/IncerTA018,DBLP:conf/mlsys/LeiWCDDW19,DBLP:conf/iclr/BaoHZS022}. For instance, Qi et al. propose searching text perturbations with theoretical guarantees on attackability by assuming the non-negativity of models \cite{DBLP:conf/mlsys/LeiWCDDW19}, which produce attacks counting on submodular optimization \cite{DBLP:conf/kdd/WangHBSML020}. Indeed, the non-negativity of models leads to binary monotonic classification (without involving the outer minimization mentioned above), which circumvents any attack that utilizes either feature addition or feature removal based perturbations, but not both \cite{DBLP:conf/codaspy/IncerTA018,DBLP:conf/uss/Chen0SJ20}. This type of classifiers tend to sacrifice detection accuracy notably \cite{pierazzi2019intriguing}. In order to relax this overly restrictive assumption, a recent study \cite{DBLP:conf/iclr/BaoHZS022} resorts to the theory of weakly submodular optimization, which necessitates a concave and smooth model. However, modern ML architectures (e.g., deep neural networks) may not have a built-in concavity. Moreover, these models are not geared toward malware detection or adversarial training. 
From the domain of image processing, pioneers propose utilizing smooth ML models \cite{DBLP:conf/iclr/SinhaND18,pmlr-v97-wang19i,DBLP:conf/cvpr/JiaZW0WC22}, because specific distance metrics (e.g., $\ell_2$ norm) can be incorporated to shape the loss landscape, leading to local concavity w.r.t. the input and thus easing the inner maximization. 
Furthermore, smoothness benefits the convergence of the outer minimization \cite{DBLP:conf/iclr/SinhaND18}. Because the proposed metrics are geared toward continuous input, they may not be suitable for software samples that are inherently discrete. Worst yet, semantics-preserving adversarial malware examples are not necessarily generated by small perturbations \cite{al2018adversarial,pierazzi2019intriguing}.

\smallskip
\noindent{\bf Our Contributions}. In this paper, we investigate adversarial training methods for malware detection by tackling three limitations of existing methods as follows. (i) We tackle the {\em robustness gap} by relaxing the constraint of ``side-effect'' features in training, and demonstrating that the resultant feature-space model can defend against practical attacks. (ii) We address the issue of {\em adversarial training without convergence guarantee} by learning convex measurements from data for quantifying distribution-wise perturbations, which regard examples falling outside of the underlying distribution as adversaries.
%\footnote{maybe better to treat them as ``unknown''?} 
In this way, the inner maximizer has to bypass the malware detector and the newly introduced adversary detector, leading to a constrained optimization problem whose Lagrangian relaxation for smooth malware detectors can be concave. Consequently, the smoothness benefits the convergence of gradient-based outer minimization \cite{DBLP:conf/iclr/SinhaND18}. (iii) We address the {\em incapability of rigorously resisting a range of attacks} by mixing multiple types of gradient-based attack methods to approximate the optimal attack, which is used to implement adversarial training while enjoying the optimization convergence mentioned in (ii). Our contributions are summarized as follows:
\begin{itemize}[leftmargin=*]
\item {\bf Adversarial training with formal treatment}. We propose a new adversarial training framework, dubbed \underline{P}rincipled \underline{A}dversarial Malware \underline{D}etection (PAD). PAD extends the malware detector with a customized adversary detector, where the customization is the convex distribution-wise measurement. For smooth models, PAD benefits convergence guarantees for adversarial training, resulting in provable robustness. 

\item {\bf Robustness improvement}. We establish a PAD model by combining a Deep Neural Network (DNN) based malware detector and an input convex neural network based adversary detector. Furthermore, we enhance the model by leveraging adversarial training to incorporate a new mixture of attacks, termed \underline{S}tepwise \underline{M}ixture of \underline{A}ttacks, leading to the defense model dubbed PAD-SMA. Theoretical analysis shows the robustness of PAD-SMA, including attackability of inner maximization and convergence of outer minimization.

\item {\bf Experimental validation}. 
We compare PAD-SMA with seven defenses proposed in the literature via the widely-used Drebin \cite{Daniel:NDSS} and Malscan \cite{DBLP:conf/kbse/WuLZYZ019} malware datasets while considering a spectrum of attack methods, ranging from no attacks, 13 oblivious attacks, to 18 adaptive attacks. 
Experimental results show that PAD-SMA significantly outperforms the other defenses, by slightly sacrificing the detection accuracy when there are no adversarial attacks.
Specifically, PAD-SMA thwarts a broad range of attacks effectively, exhibiting an accuracy $\geq 81.18\%$ under 30 attacks on Drebin and an accuracy $\geq 83.45\%$ under 27 attacks on Malscan, except for the Mimicry attack guided by multiple (e.g., 30 on Drebin or 10 on Malscan) benign software samples \cite{rndic_laskov,li2020adversarial}; it outperforms some anti-malware scanners (e.g., Symantec, Comodo), matches with some others (e.g., Microsoft), but falls behind Avira and ESET-NOD32 in terms of defense against adversarial malware examples (while noting that the attacker knows our features but not that of the scanners.)

\end{itemize}

To the best of our knowledge, this is the first principled adversarial training framework for malware detection. We have made our code publicly available at \url{https://github.com/deqangss/pad4amd}.

\noindent
\textbf{Paper outline}. 
Section \ref{sec:background-knowledge} reviews some background knowledge. Section \ref{sec:pad} describes the framework of principled adversarial malware detection. Section \ref{sec:methodology} presents a defense method instantiated from the framework. Section \ref{sec:theory-ana} analyzes the proposed method. Section \ref{sec:exp} presents our experiments
and results. Section \ref{sec:related-work} discusses related prior studies. Section \ref{sec:conclusion} concludes the paper.

\section{Background Knowledge} \label{sec:background-knowledge}

\noindent{\bf Notations}. The main notations are summarized as follows:
\begin{itemize}[leftmargin=*]
    \item {\bf Input space}: Let $\mathcal{Z}$ be the software space (i.e., problem space), and $z\in\mathcal{Z}$ be an example.

    \item {\bf Feature extraction}: Let $\phi:\mathcal{Z}\to\mathcal{X}$ be a hand-crafted feature extraction function, where $\mathcal{X}\subset\mathbb{R}^d$ is a discrete space and $d$ is the number of dimensions.
    
    \item {\bf Malware detector}: Let $f:\mathcal{Z}\to\mathcal{Y}$ map $z\in \mathcal{Z}$ to label space $\mathcal{Y}=\{0,1\}$, where ``0'' (``1'') means software example $z$ is benign (malicious). 
    
    \item {\bf Adversary detector}: Let $g:\mathcal{Z}\to\mathbb{R}$ map $z\in\mathcal{Z}$ to a real-valued confidence score such that $g(z)>\tau$ means $z$ is adversarial and non-adversarial otherwise, where $\tau$ is a pre-determined threshold. 

    \item {\bf Learning model}: We extend malware detector $f$ with a secondary detector $g$ for identifying adversarial examples. Suppose $f$ uses an ML model $\varphi_\theta:\mathcal{X}\to\mathcal{Y}$ with $f(\cdot)=\varphi_\theta(\phi(\cdot))$ and $g$ uses an ML model $\psi_\vartheta$ with $g(\cdot)=\psi_\vartheta(\phi(\cdot))$, where $\theta,\vartheta$ are learnable parameter sets.

    \item {\bf Loss function for model}: $\ce(\theta,\mathbf{x},y)$ and $\de(\vartheta,\mathbf{x})$ are the loss functions for learning models $\varphi_\theta$ and $\psi_\vartheta$, respectively.

    \item {\bf Criterion for attack}: Let $\mathcal{J}(\mathbf{x})$ justify an adversarial example, which is based on $\ce$ or a combination of $\ce$ and $\psi_\vartheta$ depending on the context. 
    
    \item {\bf Training dataset}: Let $D_z$ denote the training dataset that contains example-label pairs. Furthermore, we have $D_\mathbf{x}=\{(\mathbf{x},y):\mathbf{x}=\phi(z),(z,y)\in D_z\}$ in the feature space, which is sampled from a unknown distribution $\mathbb{P}$.
    
    \item {\bf Adversarial example}: Adversarial malware example $z'=z+\delta_z$ misleads $f$ and $g$ simultaneously (if $g$ is present), where $\delta_z$ is a set of manipulations (e.g., string injection). Correspondingly, let $\mathbf{x}'=\phi(z')$ denote the adversarial example in the feature space with $\delta_\mathbf{x}=\mathbf{x}'-\mathbf{x}$.
\end{itemize}
 
\subsection{ML-based Malware \& Adversary Detection}

We treat malware detection as binary classification. In addition, an auxiliary ML model is used to detect adversarial examples \cite{carlini2017adversarial,DBLP:conf/eurosp/SperlKCLB20,DBLP:journals/corr/abs-2106-15023}.
\begin{figure}[!htbp]
	\centering
	\includegraphics[width=0.32\textwidth]{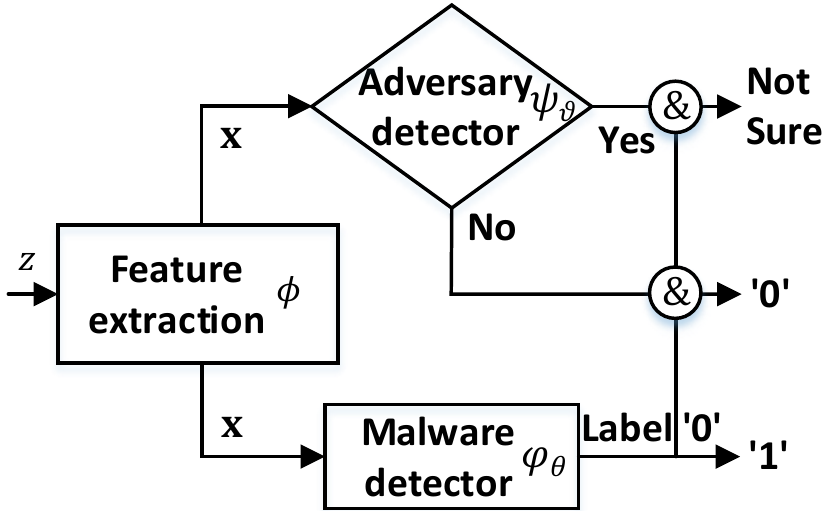}
	\caption{Integrated malware and adversary detectors.}
	\label{fig:fw-wf}
\end{figure}

Fig.\ref{fig:fw-wf} illustrates the workflow of integrated malware and adversary detectors. Formally, given an example-label pair $(z,y)$, an malware detector $f=\varphi_\theta\circ\phi$, and an adversary detector $g=\psi_\vartheta\circ\phi$, the prediction is 
\begin{equation}
	\text{predict}(z) = \left \{
\begin{aligned}
&f(z), && \text{if}~ g(z)\leq\tau \\
&1, && \text{if}~ (g(z)>\tau) \land (f(z)=1) \\
&\text{not sure}, && \text{if}~(g(z)>\tau) \land (f(z)=0). %g(z)<\tau
\end{aligned} \right. \label{eq:predict}
\end{equation} 
Intuitively, $g$ ``protects'' $f$ against $z$ when $g(z)>\tau$ and $f(z)=1$; ``not sure'' 
%option 
abstains $f$ from classification, calling for further analysis.
%, akin to the false positives (malware is the positive class).
Hence, a small portion of normal (i.e., unperturbed) examples will be flagged by $g$. Detectors $\varphi_\theta$ and $\psi_\vartheta$ are learned from training dataset $D_\mathbf{x}$ by minimizing:
\begin{eqnarray}
	\min_{\theta,\vartheta}\mathbb{E}_{(z,y)\in D_\mathbf{x}}\left[\ce(\theta,\mathbf{x},y)+\de(\vartheta,\mathbf{x})\right], \label{eq:loss1}
\end{eqnarray}
where $\ce$ is the loss for learning $\varphi_\theta$ (e.g., cross-entropy \cite{lecun2015deep}) and $\de$ is for learning $\psi_\vartheta$ (which is specified according to the downstream un-supervised task).

\subsection{Evasion Attacks} \label{sec:aea}

The evasion attack can be manifested in both the problem space and the feature space \cite{pierazzi2019intriguing,li2020adversarial}. In the problem space, an attacker perturbs a malware example $z$ to $z'$ to evade both $f$ and $g$ (if $g$ is present). Consequently, we have $\mathbf{x}=\phi(z)$ and  $\mathbf{x}'=\phi(z')$ in the feature space. Owing to the non-differentiable nature of $\phi$,
previous studies suggest $\mathbf{x}'$ obeys a ``box'' constraint $\underline{\mathbf{u}}\preceq\mathbf{x}'\preceq\overline{\mathbf{u}}$ (i.e., $\mathbf{x}'\in[\underline{\mathbf{u}}, \overline{\mathbf{u}}]$) corresponding to file manipulations, where ``$\preceq$'' is the element-wise ``no bigger than'' relation between vectors \cite{al2018adversarial,demontis2018intriguing,li2020adversarial}. The evasion attack in the feature space can be described as: 
\begin{align}
    \mathbf{x'}=\mathbf{x} +& \delta_\mathbf{x}, \\\label{eq:evs-obj}
    \text{s.t.}~ (\varphi_\theta(\mathbf{x}')=0) \land (\psi_\vartheta(\mathbf{x}')\leq\tau) &\land (\mathbf{x}'\in \mathcal{X}) \land (\mathbf{x'}\in[\underline{\mathbf{u}},\overline{\mathbf{u}}]). \nonumber
\end{align}
%Because $\varphi(\mathbf{x}')=0$ is non-differentiable, an alternative is to maximize a criterion function (e.g., $\ce(\theta,\mathbf{x}',1)$) \cite{goodfellow6572explaining}. 
Since $\psi_\vartheta$ may not be present, in what follows we review former attack methods as they are, introduce the existing strategies to target both $\varphi_\theta$ and $\psi_\vartheta$, and bring in the current inverse-mapping solutions (i.e., mapping feature perturbations to the problem space; see $\phi^{-1}$ in Figure \ref{fig:limit}).

\subsubsection{Evasion Attack Methods} \label{sec:attack-method}

\noindent{\bf Mimicry attack}. % This type of attack is showed by several studies. 
A mimicry attacker \cite{Biggio:Evasion,rndic_laskov,7917369} perturbs a malware example to mimic a benign application as much as possible. The attacker does not need to know the internal knowledge of models, but can query them. In such case, the attacker uses $N_{ben}~(N_{ben}\geq 1)$ benign examples separately to guide manipulation, resulting in $N_{ben}$ perturbed examples, of which the one bypassing the victim is used.

% \smallskip
\noindent{\bf Grosse attack}. This attack \cite{grosse2017adversarial} perturbs ``sensitive'' features to evade detection,
%achieve $\varphi_\theta(\mathbf{x})=0$. The 
where sensitivity is quantified by the gradients of the DNN's \emph{softmax} output with respect to the input. A larger gradient value means higher sensitivity. This attack adds features to an original example.

% \smallskip
\noindent{\bf FGSM attack}. This attack is introduced in the context of image classification \cite{goodfellow6572explaining} and later adapted to malware detection \cite{YeFOSINT-SI-2018,al2018adversarial}. It perturbs a feature vector $\mathbf{x}$ in the direction of the $\ell_\infty$ norm of gradients (i.e., sign operation) of the loss function with respect to the input:
$$\mathbf{x}' = \round\left(\proj_{[\underline{\mathbf{u}},\overline{\mathbf{u}}]}\left(\mathbf{x} + \varepsilon\cdot \sign(\nabla_\mathbf{x}\ce(\theta,\mathbf{x}, 1))\right)\right),$$ 
where $\varepsilon>0$ is the step size, $\proj_{[\underline{\mathbf{u}},\overline{\mathbf{u}}]}$ projects an input into $[\underline{\mathbf{u}},\overline{\mathbf{u}}]$, and $\round$ is an element-wise operation which returns an integer-valued vector. 

% \smallskip
\noindent{\bf Bit Gradient Ascent (BGA)} and \textbf{Bit Coordinate Ascent (BCA) attacks}. Both attacks \cite{al2018adversarial} iterate multiple times. In each iteration, BGA increases the feature value from `0' to `1' (i.e., adding a feature) if the corresponding partial derivative of the loss function with respect to the input is greater than or equal to the gradient's $\ell_2$ norm divided by $\sqrt{d}$, where $d$ is the input dimension. By contrast, at each iteration, BCA flips the value of the feature from `0' to `1' corresponding to the max gradient of the loss function with respect to the input. Technically speaking, given a malware instance-label pair $(\mathbf{x},y)$, the attacker needs to solve 
\begin{equation}
\max \limits_{\mathbf{x}' \in [\underline{\mathbf{u}},\overline{\mathbf{u}}]}\ce(\theta,\mathbf{x}', 1)~\text{s.t.,}~ \mathbf{x}'\in\mathcal{X} \nonumber.
\end{equation}

% \smallskip
\noindent{\bf Projected Gradient Descent (PGD) attack}. It is proposed in the image classification context \cite{madry2017towards} and adapted to malware detection by accommodating the discrete input space \cite{DBLP:journals/tnse/LiLYX21}. The attack permits both feature addition and removal while retaining malicious functionalities, giving more freedom to the attacker. It finds perturbations via an iterative process with the initial perturbation as a zero vector:
\begin{equation}
\delta^{(t+1)}_\mathbf{x}=\proj_{[\underline{\mathbf{u}}-\mathbf{x},\overline{\mathbf{u}}-\mathbf{x}]}\left(\delta^{(t)}_\mathbf{x}+\alpha\nabla_{\delta_\mathbf{x}}\ce(\theta,\mathbf{x} + \delta^{(t)}_\mathbf{x}, 1)\right) \label{eq:pgd-basic}
\end{equation}
where $t$ is the iteration, $\alpha>0$ is the step size, $\proj_{[\underline{\mathbf{u}}-\mathbf{x},\overline{\mathbf{u}}-\mathbf{x}]}$ projects perturbations into the predetermined space $[\underline{\mathbf{u}}-\mathbf{x},\overline{\mathbf{u}}-\mathbf{x}]$, and $\nabla_{\delta_\mathbf{x}}$ denotes the derivative of loss function $\ce$ with respect to $\delta_\mathbf{x}^{(t)}$. Since the derivative may be too small to make the attack progress, researchers normalize $\nabla_{\delta_\mathbf{x}}\ce$ in the direction of $\ell_1$, $\ell_2$, or $\ell_\infty$ norm \cite{madry2017towards,li2020enhancing}:
\begin{equation}
\mathbf{e}_p=\argmax\limits_{\|\mathbf{e}\|_p= 1}\langle\nabla_{\delta_\mathbf{x}}\ce(\theta,\mathbf{x} + \delta^{(t)}_\mathbf{x}, 1),\mathbf{e}\rangle, \nonumber
\end{equation}
where $\mathbf{e}_p$ is the direction of interest, $\langle\cdot,\cdot\rangle$ denotes the inner product, and $p=1,2,\infty$. Adjusting $p$ leads to PGD-$\ell_1$, PGD-$\ell_2$, and PGD-$\ell_\infty$ attacks, respectively. After the loop, an extra operation is conducted to discretize the real-valued vector. For example, $\round(\mathbf{a})$ returns the vector closest to $\mathbf{a}$ in terms of $\ell_1$ norm distance.

% \smallskip
\noindent{\bf Mixture of Attacks (MA)}. This attack \cite{li2020adversarial} organizes a mixture of attack methods upon a set of manipulations as large as possible. Two MA strategies are used: the ``max'' strategy selects the adversarial example generated by several attacks via maximizing a criterion (e.g., classifier's loss function $\ce$); the iterating ``max'' strategy puts the resulting example from the last iteration as the new starting point, where the initial point is $\mathbf{x}$. The iteration can promote attack effectiveness because of the non-concave ML model.

\subsubsection{Oblivious vs. Adaptive Attacks} \label{sec:bk-adaptive-attack}

The attacks mentioned above do not consider the adversary detector $g$, meaning that they degrade to oblivious attacks when $g$ is present and would be less effective. By contrast, an adaptive attacker is conscious of the presence of $g(\cdot)=\psi_\vartheta(\phi_\theta(\cdot))$, leading to an additional constraint $\psi_\vartheta(\mathbf{x}')\leq \tau$ for a given feature representation vector $\mathbf{x}$:
\begin{equation}
\max\limits_{\mathbf{x'}\in[\mathbf{\underline{u}}, \mathbf{\overline{u}}]}\ce(\theta, \mathbf{x}', 1)~~\text{s.t.,}~~(\psi_\vartheta(\mathbf{x}')\leq \tau) \land (\mathbf{x}'\in\mathcal{X}),
\end{equation}
where we substitute $\varphi(\mathbf{x}')=0$ with maximizing $\ce(\theta,\mathbf{x}',1)$ owing to the aforementioned issue of non-differentiability.
% Furthermore, 

However, $\psi_\vartheta$ may not be affine (e.g., linear transformation), meaning that the effective projection strategies used in PGD are not applicable anymore. In order to deal with $\psi_\vartheta(\mathbf{x}')\leq \tau$, researchers suggest three approaches: (i) Use gradient-based methods to cope with
\begin{equation}
\max\limits_{\mathbf{x'}\in[\mathbf{\underline{u}}, \mathbf{\overline{u}}]}[\ce(\theta, \mathbf{x}', 1)-\lambda\psi_\vartheta( \mathbf{x}')], \label{eq:lag}
\end{equation}
where $\lambda\geq0$ is a penalty factor for modulating the importance between the two items \cite{DBLP:journals/corr/abs-2106-15023}. 
(ii) Maximize $\ce(\theta, \mathbf{x}', 1)$ and $-\psi_\vartheta(\mathbf{x}')$ alternatively as it is notoriously difficult to set $\lambda$ properly \cite{carlini2017adversarial}. (iii) Maximize $\ce(\theta, \mathbf{x}', 1)$ and $-\psi_\vartheta(\mathbf{x}')$ in an orthogonal manner \cite{DBLP:journals/corr/abs-2106-15023}, where ``orthogonal'' means eliminating the mutual interaction between $\ce$ and $\psi$ from the geometrical perspective. For example, the attack perturbs $\mathbf{x}$ in the direction orthogonal to the direction of the gradients of $-\psi_\vartheta$, which is in the direction of the gradients of $\ce$, to make it evade $\varphi_\theta$ but not react $\psi_\vartheta$. Likewise, the attack alters the orthogonal direction to evade $\psi_\vartheta$ but not react $\varphi_\theta$.

\subsubsection{The Inverse Feature-Mapping Problem} \label{sec:bk-ps-attack}

There is a gap between the feature space and the problem (i.e., software) space. Since feature extraction $\phi$ is non-differentiable, gradient-based methods cannot produce end-to-end adversarial examples. Moreover, $\phi^{-1}$ cannot be derived analytically due to ``side-effect'' features, which cause a non-bijective $\phi$ \cite{pierazzi2019intriguing}. 

To fill the gap, Srndic and Laskov \cite{rndic_laskov} propose directly mapping the perturbation vector $\delta_\mathbf{x}$ to the problem space, leading to $\phi(\tilde{\phi}^{-1}(\mathbf{x}'))\neq\mathbf{x}'$, where $\tilde{\phi}^{-1}$ is an approximation of  $\phi^{-1}$. Nevertheless, experiments demonstrate that the attacks can evade anti-malware scanners.
Li and Li \cite{li2020adversarial} leverage this strategy to produce adversarial Android examples. Researchers also attempt to align $\delta_z$ with $\delta_\mathbf{x}$ as much as possible. For example, Pierazzi et al. \cite{pierazzi2019intriguing} collect a set of manipulations from gadgets of benign applications and implement ones that mostly align with the gradients of the loss function with respect to the input. Zhao et al. \cite{DBLP:conf/ccs/ZhaoZZZZLYYL21} propose incorporating gradient-based methods with Reinforcement Learning (RL), of which an RL-based model assists in obtaining manipulations in the problem space under the guidance of gradient information. In addition, black-box attack methods (without knowing the internals of the detector) directly manipulate malware examples, which avoids the inverse feature-mapping procedure \cite{DBLP:journals/tifs/DemetrioBLRA21}. 

In this paper, we use an approximate $\tilde{\phi}^{-1}$ (implementation details are deferred to the supplementary material).
This strategy relatively eases the attack implementation and besides, our preliminary experiments show the ``side-effect'' features cannot decline the attack effectiveness notably in the refined Drebin feature space \cite{Daniel:NDSS}.

\subsection{Adversarial Training}

Adversarial training augments training dataset with adversarial examples by solving a min-max optimization problem \cite{xu2014evasion,DBLP:conf/eisic/ChenYB17,grosse2017adversarial,madry2017towards,DBLP:conf/iclr/TramerKPGBM18,al2018adversarial}, as shown in Figure \ref{fig:limit}. The inner maximization looks for adversarial examples, while the outer minimization optimizes the model parameters upon the updated training dataset. Formally, given the training dataset $D_\mathbf{x}$, we have
\begin{align}
\min\limits_{\theta}\mathbb{E}_{(\mathbf{x},y)\in D_\mathbf{x}}&\left[\ce(\theta,\mathbf{x},y)+\beta\max\limits_{\mathbf{x}'\in[\mathbf{\underline{u}}, \mathbf{\overline{u}}]}\ce(\theta,\mathbf{x}',1)\right], \label{eq:min-max-basic}\\ ~\text{s.t.,}~&(\mathbf{x}'=\mathbf{x}+\delta_\mathbf{x}) \land (\mathbf{x}'\in\mathcal{X}) \nonumber
    % +\ce(\theta,\mathbf{x},y)
\end{align}
where $\beta\geq0$ is used to balance between detection accuracy and robustness, while noting that only malware representations play a role in the inner maximization.

However, Owing to the NP-hard nature of searching discrete perturbations \cite{DBLP:conf/iclr/SinhaND18}, the adversarial training methods incorporate the (approximate) optimal attack without convergence guaranteed  \cite{al2018adversarial,DBLP:journals/tnse/LiLYX21}, making their robustness questionable. For example, Al-Dujaili et al. \cite{al2018adversarial} approximate the inner maximization via four types of attack algorithms, while showing that a hardened model cannot mitigate the attacks that are absent in the training phase. Furthermore, a mixture of attacks is used to instantiate the framework of adversarial training \cite{li2020adversarial}. Though the enhanced model can resist a range of attacks, it is still vulnerable to a mixture of attacks with iterative ``max'' strategy (more iterations are used, see Section \ref{sec:attack-method}). Thereby, it remains a question of rigorously uncovering the robustness of adversarial training. 
\section{The PAD Framework} \label{sec:pad}

PAD aims to reshape adversarial training by rendering the inner maximization solvable analytically, with the establishment of a concave loss w.r.t. the input.
The core idea is a learnable convex distance metric, with which distribution-wise perturbations can be measured, leading to a constraint attack problem, whose Lagrange relaxation is concave (owing to the maximization) at reasonable circumstances.

\subsection{Threat Model and Design Objective}

\noindent{\bf Threat model}. Given a malware example $z$, malware detector $f$, and adversary detector $g$ (if $g$ exists), an attacker modifies $z$ by searching for a set of non-functional instructions $\delta_z$ upon knowledge of detectors. 
Guided by Kerckhoff's principle that defense should not count on ``security by obscurity'' \cite{pierazzi2019intriguing}, we consider {\em white-box} attacks, meaning that the attacker has full knowledge of $f$ and $g$. For assessing robustness of defense models, we use {\em grey-box} attacks where the attacker knows $f$ but not $g$ (i.e., oblivious attack \cite{pang2018towards}), or knows features used by $f$ and $g$.

% \smallskip
\noindent{\bf Design Objective}. As aforementioned, PAD is rooted in adversarial training. We propose incorporating $f$ with an adversary detector $g(\cdot)=\psi_\vartheta(\phi(\cdot))$, where $\psi_\vartheta$ is the convex measurement. To this end, given a malware instance-label pair $(\mathbf{x},y)$ where $\mathbf{x}=\phi(z)$ and $y=1$, we mislead both $\phi_\theta$ and $\psi_\vartheta$ by perturbing $\mathbf{x}$ to $\mathbf{x}'$, upon which we optimize model parameters. Formally, PAD uses objective
\begin{subequations}
\begin{equation}
    \begin{split}
\min\limits_{\theta,\vartheta}\mathbb{E}_{(\mathbf{x},y)\in D_\mathbf{x}}&\Bigr[\ce(\theta,\mathbf{x},y)+\de(\vartheta,\mathbf{x})\\
        &+\beta_1\ce(\theta,\mathbf{x}',1)+\beta_2\de(\vartheta,\mathbf{x}')\Bigr], \\
    \end{split} \label{eq:pad-min}
\end{equation}
where
\begin{equation} 
\begin{split}
\mathbf{x}':=\max\limits_{\mathbf{x}'\in[\mathbf{\underline{u}},\mathbf{\overline{u}}]}&\left[\ce(\theta,\mathbf{x}',1)-\lambda\psi_\vartheta(\mathbf{x}')\right], \\
    \text{s.t.}~&(\mathbf{x}+\delta_\mathbf{x}=\mathbf{x}')\land(\mathbf{x}'\in\mathcal{X}), 
    \end{split}  \label{eq:pad-max}
\end{equation}
\end{subequations}
$\beta_1$ and $\beta_2$ weight the robustness against $\mathbf{x}'$, and $\lambda\geq 0$ is a penalty factor. This formulation has three merits:

\begin{itemize}
\item[({\bf i})] {\bf Manipulations in the feature space}: Eq.\eqref{eq:pad-max} says that we can search feature perturbations without doing inverse-feature mapping, implying shorter training time. The remaining issue is %whether it will trigger concern about 
whether the attained robustness can propagate to the problem space or not; we will answer this affirmatively later (Section \ref{sec:design-rat}).

\item[({\bf ii})] {\bf Box-constraint manipulation}: Eq.\eqref{eq:pad-max} says the attacker can search $\mathbf{x}' \in [\mathbf{\underline{u}},\mathbf{\overline{u}}]$ without considering any norm-type constraints, meaning that the defender should resist semantics-based attacks rather than small perturbations.

\item[({\bf iii})] {\bf Continuous perturbation may be enough}: It is NP-hard to search optimal discrete perturbations even for attacking linear models \cite{DBLP:conf/iclr/SinhaND18}. Eq.\eqref{eq:pad-max} contains an auxiliary detector $\psi_\vartheta$, which can treat continuous perturbations in the range of $[\mathbf{\underline{u}},\mathbf{\overline{u}}]$ as anomalies while relaxing the discrete space $\mathcal{X}$ constraint in the training phase.
\end{itemize}
The preceding formulation suggests that we can use the efficient gradient-based optimization methods to solve Eq.\ref{eq:pad-min} and Eq.\ref{eq:pad-max}. In what follows we explain this intuitively and why a smooth $\ce$ is necessary (e.g., for setting a proper $\lambda$, which is challenging as discussed in Section \ref{sec:bk-adaptive-attack}).

\subsection{Design Rationale} \label{sec:design-rat}

\noindent{\bf Bridge robustness gap}. Recall that adversarial training is performed in the feature space while adversarial malware is in the problem space. Moreover, the perturbed instance $\mathbf{x}'$ used for training may not be mapped back to any $z'\in\mathcal{Z}$, because ``side-effect'' features incur interdependent perturbations (i.e., modifying one feature would require to changing some of the others so as to preserve the functionality or semantics) \cite{xu2014evasion,pierazzi2019intriguing}. This leaves a ``seam'' for attackers when a non-bijective feature extraction $\phi$ is used. Indeed, the interdependence of features is reminiscent of the structural graph representation. This prompts us to propose using a directed graph to denote the relation: modifiable features are represented by graph nodes and their interdependencies are represented by graph edges. As a result, an asymmetrical adjacent matrix (i.e., directed graph) can be used to represent the edge information, which however shrinks the manipulations in the space of $[\mathbf{\underline{u}},\mathbf{\overline{u}}]$.

Suppose for a given malware representation $\mathbf{x}$, we can obtain the optimal adversarial example in the feature space w.r.t. criterion $\mathcal{J}$. With or without considering the adjacent matrix constraint, we get the optimum $\tilde{\mathbf{x}}^\ast,\mathbf{x}^\ast\in[\mathbf{\underline{u}},\mathbf{\overline{u}}]$ with the criterion results satisfying 
$\mathcal{J}(\tilde{\mathbf{x}}^\ast)=\ce(\theta,\tilde{\mathbf{x}}^\ast,1)-\lambda\psi_\vartheta(\tilde{\mathbf{x}}^\ast)\leq\mathcal{J}(\mathbf{x}^\ast)$.
This in turn demonstrates that if an adversarial training model can resist $\mathbf{x}^\ast$, so can $\tilde{\mathbf{x}}^\ast$ (otherwise, it contradicts the meaning of optimization).

Therefore, we relax the attack constraint related to ``side-effect'' features and conduct adversarial training in the feature space so that the robustness can propagate to the problem space, at the potential price of sacrificing the detection accuracy because more perturbations are considered.

\noindent{\bf Defense against distribution-wise perturbation}. We explain Eq.\eqref{eq:pad-max} via distributionally robust optimization \cite{DBLP:conf/iclr/SinhaND18}. We establish a point-wise metric $C(\cdot,\mathbf{x})=\max\{0, \psi_\vartheta(\cdot)-\tau\}$ to measure how far a point, say $\mathbf{x}'$, to a population, while noting that other measures are also suitable as long as they are convex and continuous. 
A large portion (e.g., 95\%) of training examples will have $\psi_\vartheta(\mathbf{x})\leq\tau$. Based on $C$, we have a Wasserstein distance \cite{villani2021topics}:
$$W(\mathbb{P}',\mathbb{P}):=\inf\limits_{\Gamma}\left\{\int C(\mathbf{x}',\mathbf{x})d\Gamma(\mathbf{x}',\mathbf{x}):\Gamma\in\prod(\mathbb{P}',\mathbb{P})\right\}$$
where $\prod(\mathbb{P}',\mathbb{P})$ is the joint distribution of $\mathbb{P}'$ and $\mathbb{P}$ with marginal as $\mathbb{P}'$ and $\mathbb{P}$ w.r.t. to the first and second argument, respectively. That is, the Wasserstein distance gets the infimum from a set of expectations. Because points $\mathbf{x},\mathbf{x}'$ are in discrete space $\mathcal{X}$, the integral form in the definition is a linear summation.
We aim to build a malware detector $f$ that can classify $\mathbf{x}'$ correctly with $\mathbf{x}'\sim\mathbb{P}'$ and $W(\mathbb{P}',\mathbb{P})\leq0$. Formally, the corresponding inner maximization is
\begin{equation}
\max\limits_{\mathbb{P}':W(\mathbb{P}',\mathbb{P})\leq0}\mathbb{E}_{\mathbf{x}'\sim\mathbb{P}'}\ce(\theta,\mathbf{x}',1). \label{eq:wd-max}
\end{equation}
It is non-trivial to tackle $W(\mathbb{P}',\mathbb{P})$ directly owing to massive vectors on $\mathcal{X}\times\mathcal{X}$. Instead, the dual problem is used:
\begin{proposition} \label{prop:framework}
Given a continuous function $\ce$, and continuous and convex distance $C(\cdot, \mathbf{x})=\max\{0, \psi_\vartheta(\cdot)-\tau\}$ with $\mathbf{x}\sim\mathbb{P}$, the dual problem of Eq.\eqref{eq:wd-max} is 
    $$\inf_{\lambda}\Bigl\{\mathbb{E}_{\mathbf{x}\sim\mathbb{P}}\max\limits_{\mathbf{x}'}(\ce(\theta,\mathbf{x}',1)-\lambda\psi_\vartheta(\mathbf{x}') + \lambda\tau):\lambda\geq 0\Bigl\},$$ where $\mathbf{x}+\delta_\mathbf{x}=\mathbf{x}'\in\mathcal{X}$, $\mathbf{x}'\sim\mathbb{P}'$ and $\psi_\vartheta(\mathbf{x}')\geq \tau$.
\end{proposition}
\noindent Its empirical version is Eq.\eqref{eq:pad-max} for fixed $\lambda$ and $\tau$, except for the constraint $[\mathbf{\underline{u}},\mathbf{\overline{u}}]$ handled by clip operation. The proposition says PAD can defend against distributional perturbations. Proof is deferred to the supplementary material. 

\noindent{\bf Concave inner maximization}. 
Given an instance-label pair $(\mathbf{x},y)$, let Taylor expansion approximate $\ce(\theta,\mathbf{x}+\delta_\mathbf{x},y)-\lambda\psi_\vartheta(\mathbf{x}+\delta_\mathbf{x})$:
\begin{align*}
\ce(\theta,\mathbf{x}+&\delta_\mathbf{x},y)-\lambda\psi_\vartheta(\mathbf{x}+\delta_\mathbf{x})\cong\ce-\lambda\psi_\vartheta\\
    +&\langle\nabla_\mathbf{x}(\ce-\lambda\psi_\vartheta),\delta_\mathbf{x}\rangle +\frac{1}{2} \delta_\mathbf{x}^\top\nabla^2_\mathbf{x}(\ce-\lambda\psi_\vartheta)\delta_\mathbf{x}.
\end{align*}
where $\ce-\lambda\psi_\vartheta$ denotes $\ce(\theta,\mathbf{x},y)-\lambda\psi_\vartheta(\mathbf{x})$ for short.
The insight is that if (i) the values of the entities in $\nabla_\mathbf{x}\ce$ are finite (i.e., smoothness \cite{DBLP:conf/iclr/SinhaND18}) when $\mathbf{x}\in[\mathbf{\underline{u}},\mathbf{\overline{u}}]$, and (ii) $\nabla_\mathbf{x}\psi_\vartheta>0$ (i.e., strongly convex), then we can make $\ce-\lambda\psi_\vartheta$ concave by tweaking $\lambda$; this eases the inner maximization.   

\begin{figure}[t]
    \centering
    \begin{tabular}{l l}
        \hspace{-10pt}\parbox{1.65in}{\includegraphics[width=1.65in,height=0.7in]{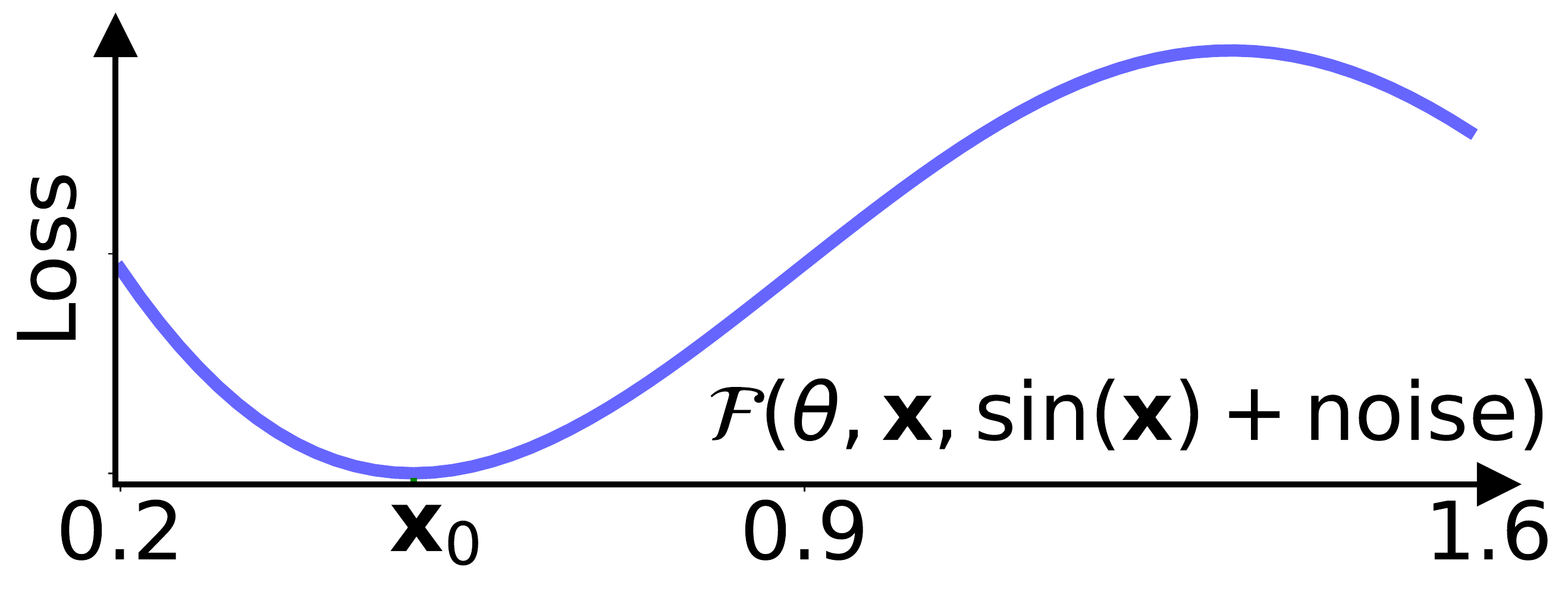}\\ \includegraphics[width=1.65in,height=0.7in]{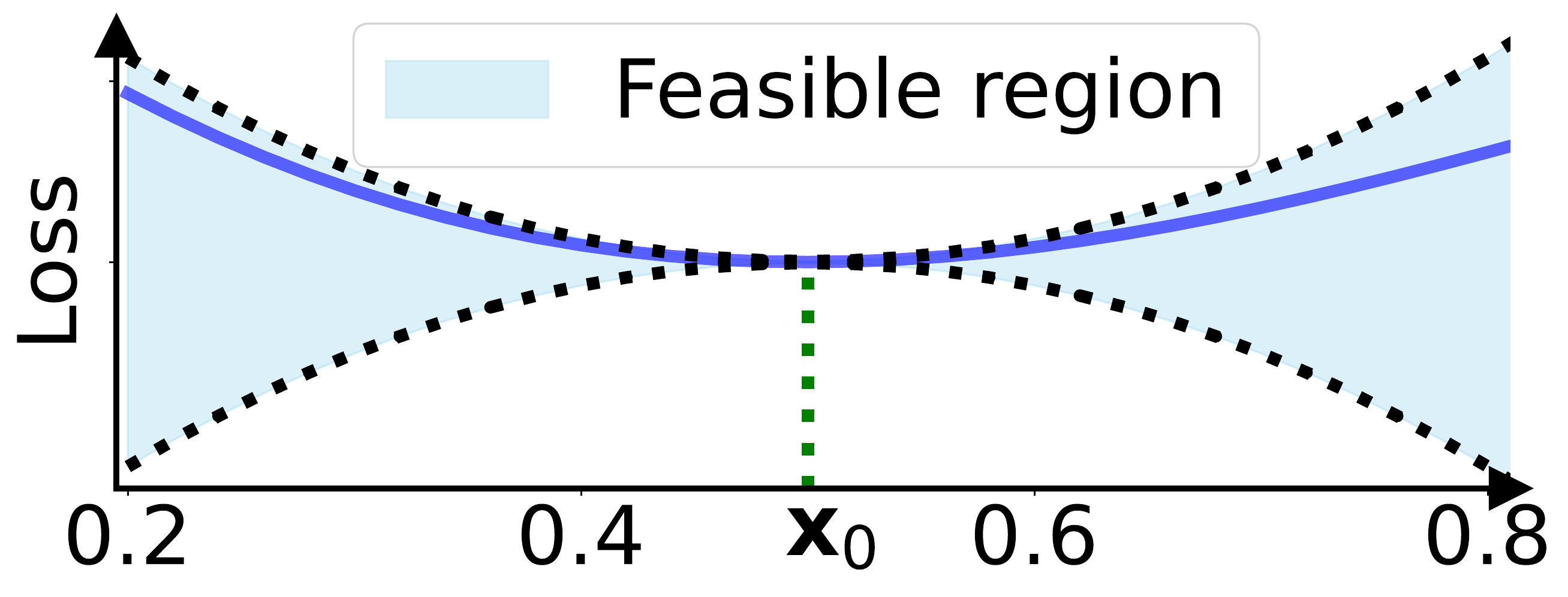}} &\hspace{-15pt} \parbox{1.7in}{\includegraphics[width=1.8in,height=1.45in]{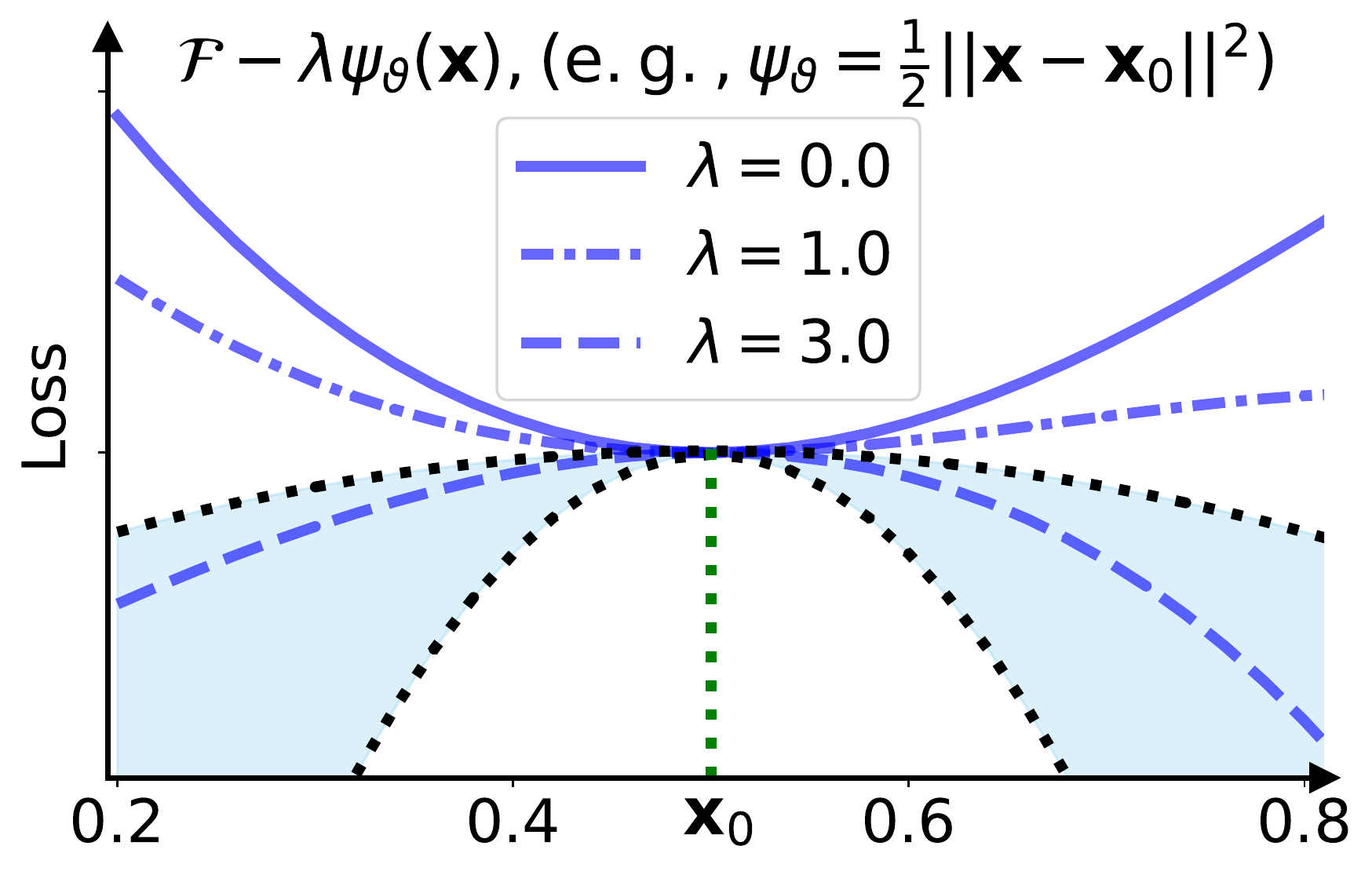}}\\
    \end{tabular}
\caption{An example showing how the loss changes under perturbations when $\ce$ is smooth (feasible region in the bottom-left figure), making $\ce-\lambda\psi_\vartheta$ strongly convex (feasible region in the rightmost figure) at $\mathbf{x}_0$ when $\lambda=3.0$.}
    \label{fig:fw-mtv}
\end{figure}

Figure \ref{fig:fw-mtv} illustrates the idea behind the design, by using a smoothed DNN model to fit the noising $\sin$ function (top-left figure). Owing to the smoothness of $\varphi_\theta$ and $\ce$ (bottom-left figure), we transform the loss function to a concave function by incorporating a convex $\psi_\vartheta$. The concavity is achieved gradually by raising $\lambda$, along with the feasible region changed, as shown in the right-hand figure.
In the course of adjusting $\lambda$, there are three possible scenarios \cite{DBLP:journals/corr/abs-2106-15023}: (i) $\lambda$ is large enough, leading to a concave inner maximization. (ii) A proper $\lambda$ may result in a linear model, which would be rare because of the difference between $\varphi_\theta$ and $\psi_\vartheta$. (iii) $\lambda$ is so small that the inner maximization is still a non-concave and nonlinear problem, which is true as former heuristic adversarial training. In summary, we propose enhancing the robustness of $f$ and $g$, which can reduce the smoothness factor of $f$ \cite{szegedyZSBEGF13,DBLP:conf/cvpr/Moosavi-Dezfooli19} and thus force the attacker to increase $\lambda$ when generating adversarial examples.

Since the interval $\mathbf{x}+\delta_\mathbf{x}\in[\mathbf{\underline{u}},\mathbf{\overline{u}}]$ relaxes the constraint on a discrete input, we can address this issue by treating continuous perturbations as anomalies, as stated earlier. Therefore, instead of heuristically searching for discrete perturbations, we directly use $\psi_\theta$ to detect continuous perturbations without using the discretization trick.

\section{Instantiating the PAD Framework}
\label{sec:methodology}

We instantiate PAD into a model and associated adversarial training algorithm. Though PAD may be applicable to any differentiable ML algorithms, we consider Deep Neural Network (DNN) based malware detection because it has been intensively investigated \cite{DBLP:journals/tnn/YuanHZL19,10.1145/3544968,kolosnjaji2018adversarial,li2020adversarial}. 

\subsection{Adjusting Malware Detector} \label{sec:method:adj-md}

PAD requires the composition of $\ce$ and $\varphi_\theta$ to be smooth. DNN consists of hierarchical layers, each of which typically has a linear mapping followed by a non-linear activation function. Most of these ingredients meet the smoothness condition, except for some activation functions (e.g., Rectified Linear Unit or ReLU \cite{lecun2015deep}) owing to non-differentiability at point zero.   
To handle non-smooth activation functions, researchers suggest using over-parameterized DNNs, which yield semi-smooth loss landscapes \cite{pmlr-v97-allen-zhu19a}. Instead of increasing learnable parameters, we replace ReLU with smooth activation functions (e.g., Exponential Linear Unit or ELU \cite{DBLP:journals/corr/ClevertUH15}). The strategy is simple in the sense that the model architecture is changed slightly and fine-tuning suffices to recover the detection accuracy. Despite this, our preliminary experiments show it slightly reduces the detection accuracy.

\subsection{Adversary Detector} \label{sec:ad-icnn}

We propose a DNN-based $g$ that is also learned from the features extracted by $\phi$. Figure \ref{fig:icnn} shows the architecture of $\psi_\vartheta$, which is an $l$-layer Input Convex Neural Network (ICNN) \cite{DBLP:conf/icml/AmosXK17}. ICNN maps an input $\mathbf{x}$ recursively via non-negative transformations, along with adding a normal transformation on $\mathbf{x}$:
$$
\mathbf{x}^{i+1}=\sigma(\bm{\vartheta}^{i}\mathbf{x}^{i} + \bm\vartheta^{i}_\mathbf{x}\mathbf{x}+{\bf b}^{i}),
$$
where $\vartheta=\{\bm\vartheta^{i},\bm\vartheta^{i}_{\mathbf x},\mathbf{b}^{i}:i=1,\ldots,l\}$, $\bm\vartheta^{i}$ is non-negative, $\bm\vartheta^{i}_{\mathbf x}$ has no such constraint, $\mathbf{x}^{1}=\mathbf{x}$,  $\bm\vartheta^{1}$ is identity matrix, and $\sigma$ is a smooth activation function (e.g., ELU or Sigmoid \cite{DBLP:journals/corr/ClevertUH15}).

\begin{figure}[!tbhp]
	\centering
	\includegraphics[width=0.44\textwidth]{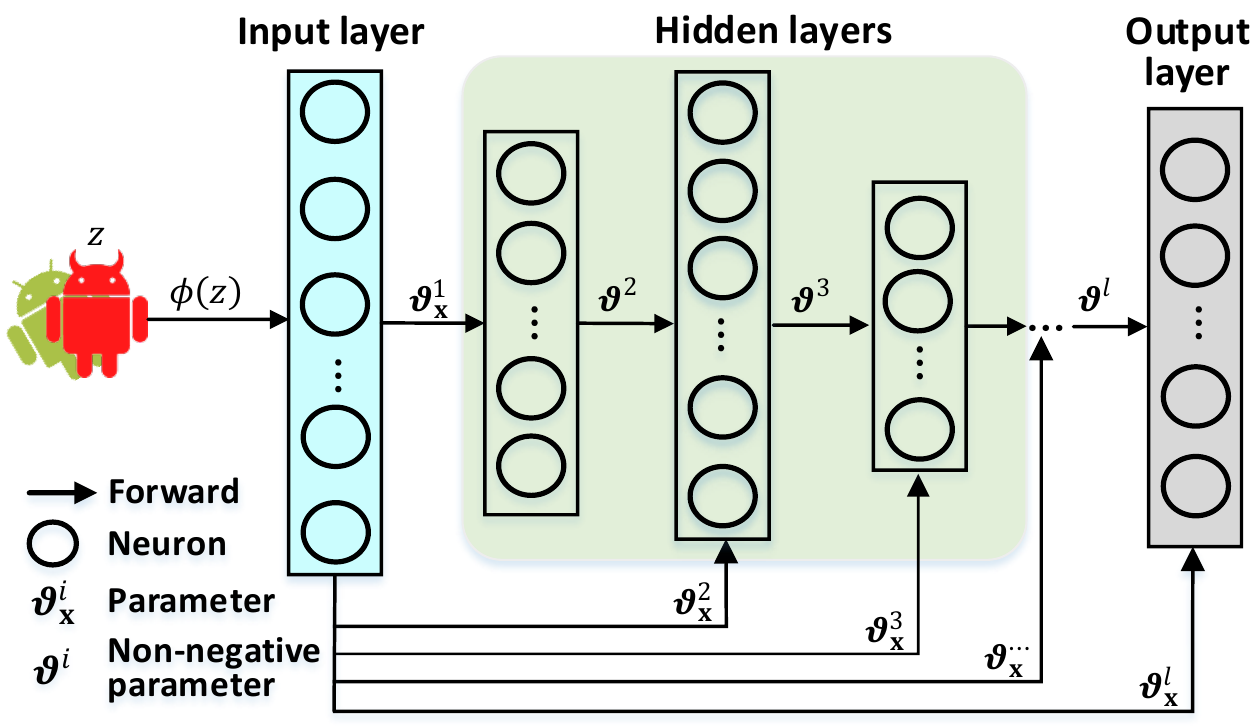}
	\caption{Architecture of an input convex neural network.}
	\label{fig:icnn}
\end{figure}

We cast the adversary detection as a one-class classification task \cite{DBLP:journals/spl/OzaP19}. In the training phase, we perturb examples in $D_\mathbf{x}$ to obtain a set of new examples $\{\mathbf{x}+\delta_\mathbf{x}:(\mathbf{x},y)\in D_\mathbf{x}$\}, where $\delta_\mathbf{x}$ is a vector of salt-and-pepper noises, meaning that at least half of elements in ${\bf x}$ are randomly selected and their values are set as their respective maximum.
Formally, given an example $\mathbf{x}^{1}\in\{\mathbf{x}:(\mathbf{x},y)\in D_\mathbf{x}\}\cup\{\mathbf{x}+\delta_\mathbf{x}:(\mathbf{x},y)\in D_\mathbf{x}\}$, the loss function $\de$ is
$$
\mathcal{G}(\vartheta,\mathbf{x}^{1})={\sf pert}\log(\psi_\vartheta(\mathbf{x}^{1})) + (1-{\sf pert})\log(1 - \psi_\vartheta(\mathbf{x}^{1})),
$$
where ${\sf pert}=0$ indicates $\mathbf{x}^{1}$ is from $D_\mathbf{x}$, and ${\sf pert}=1$ otherwise. In the test phase, we let the input pass through $\psi_\vartheta$ to perform the prediction as shown in Eq.\eqref{eq:predict}.

\subsection{Adversarial Training Algorithm} \label{sec:ata}

For the {\em inner maximization} (Eq.\ref{eq:pad-max}), we propose a mixture of PGD-$\ell_1$, PGD-$\ell_2$ and PGD-$\ell_\infty$ attacks (see Section \ref{sec:attack-method}). The attacks proceed iteratively via ``normalized'' gradients
\begin{equation}
\mathbf{e}_p=\argmax\limits_{\|\mathbf{e}\|_p=1}\langle\nabla_{\delta_\mathbf{x}}(\ce(\theta,\mathbf{x}+\delta_{\mathbf x}^{(t)},1)-\lambda\psi_\vartheta(\mathbf{x}+\delta_{\mathbf x}^{(t)})),\mathbf{e}\rangle, \label{eq:pgd-norm2}
\end{equation}
and perturbation vectors
\begin{equation}
\left\{\delta^{(t+1)}_{\mathbf{x},p}=\proj_{[\mathbf{\underline{u}}-\mathbf{x},\mathbf{\overline{u}}-\mathbf{x}]}\left(\delta^{(t)}_{\mathbf{x},p}+\alpha_p\mathbf{e}_p\right):p\in\{1,2,\infty\}\right\}, \label{eq:pgd}
\end{equation}
where a perturbation vector is chosen by the scoring rule
\begin{equation}
\begin{aligned}
\smash{\delta_\mathbf{x}^{(t+1)}=\argmax_{\delta_{\mathbf{x},p}^{(t+1)}}}\Bigr[&\ce(\theta,\round(\mathbf{x}+\delta_{\mathbf{x},p}^{(t+1)}),1)\\
    &-\lambda\psi_\vartheta(\round(\mathbf{x}+\delta_{\mathbf{x},p}^{(t+1)}))\Bigr] 
\end{aligned}\label{eq:pgd-sel}
\end{equation}
at the $t^\text{th}$ iteration. The $\round$ operation is used because our initial experiments show that it leads to better robustness. Since the goal is to select the best attack in a stepwise fashion, it is termed Stepwise Mixture of Attacks (SMA).

Note that from an attacker's perspective, there are three more steps: (i) We treat the dependencies between features as graphical edges. Since the summation of gradients can measure the importance of a group in the graph \cite{wu2019adversarial}, we accumulate the gradients of the loss function with respect to the ``side-effect'' features and use the resulting gradient to decide whether to modify these features together.
(ii) The $\round$ operation is used to discretize perturbations when the loop is terminated \cite{DBLP:journals/tnse/LiLYX21}. (iii) Map the perturbations back into the problem space.

For the {\em outer minimization} (Eq.\ref{eq:pad-min}), we leverage a Stochastic Gradient Descent (SGD) optimizer, which proceeds iteratively to find the model parameters. 
Basically, SGD samples a batch of $B$ (a positive integer) pairs $\{(\mathbf{x}_i,y_i)\}_{i=1}^B$ from $ D_\mathbf{x}$ and updates the parameters with
\begin{align}
    \theta^{(j+1)}&=&\theta^{(j)}-\gamma\nabla_\theta\frac{1}{B}\smash{\tiny \sum_{i=1}^{B}}\ce(\theta^{(j)},\mathbf{x}_i+\delta_{\mathbf{x}_i}^{(T)},y_i) ~~\text{and}~~\nonumber \\
    \vartheta^{(j+1)}&=&\vartheta^{(j)}-\gamma\nabla_\vartheta\frac{1}{B}{\tiny \sum_{i=1}^{B}}\de(\vartheta^{(j)},\mathbf{x}_i+\delta_{\mathbf{x}_i}^{(T)}), \nonumber
\end{align}
where $j$ is the iteration, $\gamma$ is the learning rate, and $\delta_{\mathbf{x}_i}^{(T)}$ is obtained from Eq.\eqref{eq:pgd-sel} with $T$ loops for perturbing $\mathbf{x}_i$. 
We optimize the model parameters by Eq.\eqref{eq:pad-min}.

\begin{algorithm}[!htbp]
\KwIn{Training set $D_z$, epoch $N$, batch size $B$, factors $\beta_1$, $\beta_2$ and $\lambda$, iteration $T$, and step size $\alpha_p$ for norm $p\in\{1,2,\infty\}$.}

Get $D_\mathbf{x}=\{(\phi(z),y):(z,y)\in D_z\}$ for the given $D_z$;

\For{$j=$ \rm ${1}$ to $N$}{
Sample a mini-batch $\{\mathbf{x}_i, y_i\}_{i=1}^B$ from $D_\mathbf{x}$;
\BlankLine

Apply salt-and-pepper noises to $\{\mathbf{x}_i\}_{i=1}^B$;

\For{$t=$ \rm ${0}$ to $T-1$}{
\For{$p\in\{1,2,\infty\}$}{
    Calculate perturbation $\delta_{\mathbf{x},p}^{(t+1)}$ by Eq.\eqref{eq:pgd-norm2} and Eq.\eqref{eq:pgd} for $\mathbf{x}\in\{\mathbf{x}_i\}_{i=1}^B$ with $y_i=1$;
}

Select $\delta_{\mathbf{x}}^{(t+1)}$ by Eq.\eqref{eq:pgd-sel};
}
\BlankLine

Calculate the adversarial training loss via Eq.\eqref{eq:pad-min};

\BlankLine
Backpropagate the errors for updating $\theta$ and $\vartheta$;
}

\caption{Adversarial training}
\label{alg:adv-train}
\end{algorithm}

Algorithm \ref{alg:adv-train} summarizes a PAD-based adversarial training by incorporating the stepwise mixture of attacks. Given a training set, we preprocess software examples and obtain their feature representations (line 1). At each epoch, we first perturb the feature representations via salt-and-pepper noises (line 4) and then generate adversarial examples with the mixture of attacks (lines 5-10). Using the union of the original examples and their perturbed variants, we learn malware detector $f$ and adversary detector $g$ (lines 11-13). 

\section{Theoretical Analysis} \label{sec:theory-ana}

We analyze effectiveness of the inner maximization and optimization convergence of the outer minimization, which together support robustness of the proposed method. % Given a pair $(\mathbf{x},y)$, we introduce a new symbol $\mathcal{J}(\mathbf{x})=\ce(\theta,\mathbf{x},\vartheta)-\lambda\psi_\vartheta(\mathbf{x})$. 
As mentioned above, we make an assumption that PAD requires smooth learning algorithms (Section \ref{sec:method:adj-md}).

\begin{assumption}[Smoothness assumption \cite{DBLP:conf/iclr/SinhaND18}] \label{assump:smooth}
	\ignore{Layer of neural network $\mathbf{F}^{(i)}~(i=1,\ldots,l)$ is $L^{(i)}_0$-Lipschitz continuous, and its Jacobian $\mathbf{J_F}^{(i)}$ is $L^{(i)}_1$-Lipschitz continuous, where $L^{(i)}_0,L^{(i)}_1$ are Lipschitz constants. Further,}
	The composition of $\ce$ and $\varphi_\theta$ meets the smoothness condition:
	\begin{align*}
	&\|\nabla_{\mathbf{x}}\ce(\theta,\mathbf{x},y) - \nabla_{\mathbf{x}}\ce(\theta,\mathbf{x}',y)\|_2\leq {\sf L}_{\mathbf{xx}}^f\|\mathbf{x}-\mathbf{x}'\|_2, \\
		&\|\nabla_{\mathbf{x}}\ce(\theta,\mathbf{x},y) - \nabla_{\mathbf{x}}\ce(\theta',\mathbf{x},y)\|_2\leq {\sf L}_{\mathbf{x}\theta}^f\|\theta-\theta'\|_2, \\
		&\|\nabla_{\theta}\ce(\theta,\mathbf{x},y) - \nabla_{\theta}\ce(\theta,\mathbf{x}',y)\|_2\leq {\sf L}_{\theta\mathbf{x}}^f\|\mathbf{x}-\mathbf{x}'\|_2,
	\end{align*}
	and
	%the composite of $\de$ and $\psi_\vartheta$ 
	$\psi_\vartheta$ 
	meets the smoothness condition:
	\begin{align*}
		&\|\nabla_{\mathbf{x}}\psi_\vartheta(\mathbf{x}) - \nabla_{\mathbf{x}}\psi_\vartheta(\mathbf{x}')\|_2\leq {\sf L}_{\mathbf{xx}}^g\|\mathbf{x}-\mathbf{x}'\|_2,\\
		&\|\nabla_{\mathbf{x}}\psi_\vartheta(\mathbf{x}) - \nabla_{\mathbf{x}}\psi_{\vartheta'}(\mathbf{x})\|_2\leq {\sf L}_{\mathbf{x}\vartheta}^g\|\vartheta-\vartheta'\|_2,
	\end{align*}
	where $\mathbf{x}'\in[\mathbf{\underline{u}},\mathbf{\overline{u}}]$ is changed from $\mathbf{x}=\phi(z)$ for a given example $z$ and ${\sf L}^{\ast}_{\ast\ast}>0$ denotes the smoothness factor ($\ast$ is the wildcard).
\end{assumption}

% \begin{assumption}[Convexity assumption] \label{assump:convex}
    % The composite of $\de$ and $\psi_\vartheta$ 
Recall that the $\psi_\vartheta$ meets the strongly-convex condition:
\begin{equation*}
     \|\nabla_{\mathbf{x}}\psi_\vartheta(\mathbf{x}) - \nabla_{\mathbf{x}}\psi_\vartheta(\mathbf{x}')\|_2\geq {\sf M}_{\mathbf{xx}}^g\|\mathbf{x}-\mathbf{x}'\|_2,
\end{equation*}
where ${\sf M}_{\mathbf{xx}}^g>0$ is the convexity factor.
% \end{assumption}

\begin{proposition} \label{lemma:conca-smooth}
Assume the smoothness assumption holds. The loss of $\ce-\lambda\psi_\vartheta$ is $(\lambda {\sf M}_{\bf xx}^g-{\sf L}_{\bf xx}^f)$-strongly concave and $(\lambda {\sf L}_{\bf xx}^g+{\sf L}_{\bf xx}^f)$-smoothness when ${\sf L}_{\bf xx}^f < \lambda {\sf M}_{\bf xx}^g$. That is
	% on the domain $[0,1]^d\times[0,1]^d$ w.r.t the node representation:
	\begin{align*}
		-\frac{\lambda {\sf L}_{\bf xx}^g+{\sf L}_{\bf xx}^f}{2}\|\mathbf{x}'-\mathbf{x}\|_2^2\leq \mathcal{L}\leq-\frac{\lambda {\sf M}_{\bf xx}^g-{\sf L}_{\bf xx}^f}{2}\|\mathbf{x}'-\mathbf{x}\|_2^2,
	\end{align*}
	where $\mathcal{L}=\ce(\theta,\mathbf{x}',y)-\lambda\psi_\vartheta(\mathbf{x}')-\ce(\theta,\mathbf{x},y)+\lambda\psi_\vartheta(\mathbf{x})-\\ \langle\nabla_\mathbf{x}(\ce-\lambda\psi_\vartheta),\delta_\mathbf{x}\rangle=\mathcal{J}(\mathbf{x}')-\mathcal{J}(\mathbf{x})-\langle\nabla_\mathbf{x}\mathcal{J}(\mathbf{x}),\delta_\mathbf{x}\rangle$.
\end{proposition}
\begin{proof}
By quadratic bounds derived from the smoothness, we have
	$-\frac{{\sf L}_{\bf xx}^f}{2}\|\mathbf{x}'-\mathbf{x}\|_2^2\leq\ce(\theta,\mathbf{x}',y)-\ce(\theta,\mathbf{x},y)-\langle\nabla_\mathbf{x}\ce,\mathbf{x}'-\mathbf{x}\rangle\leq\frac{{\sf L}_{\bf xx}^f}{2}\|\mathbf{x}'-\mathbf{x}\|_2^2$. Since $\psi_\vartheta$ is convex, we get $\psi_\vartheta(\mathbf{x}')-\psi_\vartheta(\mathbf{x})-\langle\nabla_\mathbf{x}\psi_\vartheta,\mathbf{x}'-\mathbf{x}\rangle\geq\frac{{\sf M}_{\bf xx}^g}{2}\|\mathbf{x}'-\mathbf{x}\|_2^2$. Since $\psi_\vartheta$ is smooth, we get $\psi_\vartheta(\mathbf{x}')-\psi_\vartheta(\mathbf{x})-\langle\nabla_\mathbf{x}\psi_\vartheta,\mathbf{x}'-\mathbf{x}\rangle\leq\frac{{\sf L}_{\bf xx}^g}{2}\|\mathbf{x}'-\mathbf{x}\|_2^2$. Combining these two inequalities leads to the proposition.
\end{proof}

Theorem \ref{theorem:attack} below quantifies the gap between the approximate adversarial example $\mathbf{x}'=\mathbf{x}+\delta_{\mathbf{x}}^{(T)}$ and the optimal one, denoted by $\mathbf{x}^\ast=\mathbf{x}+\delta_\mathbf{x}^\ast$. The proof is lengthy and deferred to the supplementary material.

\begin{theorem} \label{theorem:attack}
Suppose the smoothness assumption holds. If ${\sf L}_{\bf xx}^f < \lambda {\sf M}_{\bf xx}^g$, the perturbed example $\mathbf{x}'=\mathbf{x}+\delta_{\mathbf{x}}^{(T)}$ from Algorithm \ref{alg:adv-train} satisfies:
	$$
	\frac{\mathcal{J}(\mathbf{x}^\ast)-\mathcal{J}(\mathbf{x}')}{\mathcal{J}(\mathbf{x}^\ast)-\mathcal{J}(\mathbf{x})}\leq\exp(-\frac{T}{d}\cdot\frac{\lambda {\sf M}_{\bf xx}^g-{\sf L}_{\bf xx}^f}{\lambda {\sf L}_{\bf xx}^g+{\sf L}_{\bf xx}^f}),
	$$
	where $d$ is the input dimension.
 % and $\mathcal{J}(\mathbf{x}^\ast)-\mathcal{J}(\mathbf{x})$ is a constant.
\end{theorem}
% Appendix \ref{sec:appendix:theorem1}.

We now focus on the convergence of SGD when applied to the outer minimization. Without loss of generality, the following theorem is customized to the composition of $\varphi_\theta$ and $\ce$, which can be extended to the composition of $\psi_\vartheta$ and $\de$. Let  $\mathcal{H}(\theta)=\mathbb{E}_{(\mathbf{x},y)\in D_\mathbf{x}}\ce(\theta,\mathbf{x}^\ast(\theta),y)$ denote the optimal adversarial loss on the entire training dataset $D_\mathbf{x}$. 
\begin{theorem} \label{theorem:convergence}
Suppose the smoothness assumption holds. Let $\Delta=\mathcal{H}(\theta^{(0)})-\min_\theta \mathcal{H}(\theta)$. If we set the learning rate to $\gamma^{(j)}=\gamma=\min\{1/{\sf L}, \sqrt{\Delta/({\sf L}\zeta^2N)}\}$, the adversarial training satisfies
	\begin{equation}
		\frac{1}{N}\sum_{j=0}^{N}\mathbb{E}\left\|\nabla \mathcal{H}(\theta^{(j)})\right\|\leq \zeta\sqrt{8\frac{\Delta {\sf L}}{N}} + 2\hat{c},
	\end{equation} 
	where N is the number of epochs, ${\sf L}=\frac{{\sf L}^f_{\theta \mathbf{x}}(\lambda {\sf L}_{\mathbf{x}\theta}^g+{\sf L}_{\mathbf{x}\theta}^f)}{\lambda {\sf M}_{\bf xx}^g-{\sf L}_{\bf xx}^f}+{\sf L}^f_{\theta\theta}$, $\hat{c}=(\mathcal{J}(\mathbf{x}^\ast)-\mathcal{J}({\mathbf{x}}))\frac{ 2{\sf L}^f_{\theta \mathbf{x}}}{\lambda {\sf M}_{\bf xx}^g-{\sf L}_{\bf xx}^f}\exp(-\frac{T}{d}\cdot\frac{{\lambda {\sf M}_{\bf xx}^g-{\sf L}_{\bf xx}^f}}{\lambda {\sf L}_{\bf xx}^g+{\sf L}_{\bf xx}^f})$, and $\zeta$ is the variance of stochastic gradients.
	% and $\minimum$ is element-wise operation for returning the minimum element between two elements.
\end{theorem}
The proof is also deferred to the supplementary material. Theorem \ref{theorem:convergence} says that the convergence rate of the adversarial training is  $\mathcal{O}({1}/{\sqrt{N}})$. Moreover, the approximation of the inner maximization has a constant effect on the convergence because of $\hat{c}$. More importantly, attacks  achieving a lower attack effectiveness than this approximation possibly enlarge the effect and can be mitigated by this defense. 
\section{Experiments} \label{sec:exp}

We conduct experiments to validate the soundness of the proposed defense in the absence and presence of evasion attacks, by answering 4 Research Questions (RQs):
\begin{itemize}[leftmargin=*]
	\item {\bf RQ1: Effectiveness of defenses in the absence of attacks}: How effective is {PAD-SMA} when there is no attack? This is important because the defender does not know for certain whether there is an adversarial attack or not.
	\item {\bf RQ2: Robustness against oblivious attacks}: How robust is {PAD-SMA} against oblivious attacks where ``oblivious'' means the attacker is unaware of adversary detector $g$?
	\item {\bf RQ3: Robustness against adaptive attacks}: How robust is {PAD-SMA}  against adaptive attacks? 
	\item {\bf RQ4: Robustness against practical attacks}: How robust is {PAD-SMA} against attacks in the problem space? 
\end{itemize}

\noindent{\bf Datasets}. Our experiments utilize two Android malware datasets: Drebin \cite{Daniel:NDSS} and Malscan \cite{DBLP:conf/kbse/WuLZYZ019}, which are widely used in the literature. The Drebin dataset initially contains 5,560 malicious apps and features extracted from 123,453 benign apps; both were collected before the year 2013. In order to obtain the customized features,  \cite{li2020adversarial} re-collects benign apps from the Androzoo repository \cite{Allix:2016:ACM:2901739.2903508} and re-scans the collections via VirusTotal, resulting in 42,333 benign examples. This leads to the Drebin dataset used in this paper containing 5,560 malicious apps and 42,333 benign apps. Malscan \cite{DBLP:conf/kbse/WuLZYZ019} contains 11,583 malicious apps and 11,613 benign apps, spanning from 2011 to 2018.
These apps are labeled using VirusTotal \cite{VirusTotal:Online};
an app is flagged as malicious if five or more malware scanners say the app is malicious, and as benign if no malware scanners flag it as malicious. We randomly split a dataset into three disjoint sets: 60\% for training, 20\% for validation, and 20\% for testing. 

% \smallskip
\noindent{\bf Feature extraction and manipulation}. 
We use two families of features. ({\bf i}) {\bf Manifest} features, including: {\em hardware} statements (e.g., camera and GPS module) because they may incur security concerns; 
{\em permissions} because they may be abused to breach a user's privacy; implicit {\em Intents} because they are related to communications between app components (e.g., services). These features can be perturbed by injecting operations but may not be removed without undermining a program's functionality \cite{7917369,li2020adversarial}.
({\bf ii}) {\bf Classes.dex} features, including: ``restricted'' and ``dangerous'' Application Programming Interfaces (APIs), where a ``restricted'' API means that its invocation 
requires declaring the corresponding permissions and ``dangerous'' APIs include the ones related to Java reflection usage (e.g., {\tt getClass}, {\tt getMethod}, {\tt getField}), encryption usage (e.g., {\tt javax.crypto}, {\tt Crypto.Cipher}), the explicit intent indication (e.g., {\tt setDataAndType}, {\tt setFlags}, {\tt addFlags}), dynamic code loading (e.g., {\tt DexClassLoader}, {\tt System.loadLibrary}), and low-level command execution (e.g., {\tt Runtime.getRuntime.exec}). These APIs can be injected along with dead codes \cite{pierazzi2019intriguing}. Note that APIs with the {\tt public} modifier can be hidden via Java reflection \cite{li2020adversarial}, which involves reflection-related APIs used by our detector, referred to as ``side-effect'' features as mentioned above. These features may benefit the defender.

We exclude some features. For {\bf manifest} features (e.g., package name, {\em activities}, {\em services}, {\em provider}, and {\em receiver}), they can be injected or renamed \cite{aamo:Online,li2020adversarial}. For {\bf Classes.dex} features, existing manipulations include {\em string} (e.g., IP address) injection/encryption \cite{7917369,li2020adversarial}, {\em public} or {\em static} API calls hidden by Java reflection  \cite{aamo:Online,li2020adversarial}, Function Call Graph (FCG) addition and rewiring \cite{aonzo2020obfuscapk}, anti-data flow obfuscation \cite{jung:avpass-bh}, and control flow obfuscation (by using arithmetic branches) \cite{aamo:Online}. For {\bf other types} of features, app signatures can be re-signed \cite{aamo:Online}; native libraries can be modified by Executable and Linkable Format (ELF) section-wise addition, ELF section appending, and instruction substitution \cite{lief:Online}. 

We use Androguard, a reverse engineering toolkit  \cite{Androguard:Online}, to extract features. We apply a binary feature vector to denote an app, where ``1'' means a feature is present and ``0'' otherwise. The 10,000 top-frequency features are used.

% \smallskip
\noindent{\bf Defenses that are considered for comparison purposes}. 
We consider 8 representative defenses:
\begin{itemize}[leftmargin=*,align=left]
    \item {\bf DNN} \cite{grosse2017adversarial}: %{\em\underline{D}}eep {\em\underline{N}}eural {\em\underline{N}}etwork 
    DNN based malware detector with no defensive hardening, which serves as the baseline;
    
    \item {\bf AT-rFGSM$^k$} \cite{al2018adversarial}: DNN-based malware detector hardened by {\em\underline{A}}dversarial {\em\underline{T}}raining with the {\em\underline{r}}andomized $\round$ operation enabled {\em\underline{FGSM}}$^k$ attack (AT-rFGSM$^k$);
    
    \item {\bf AT-MaxMA} \cite{li2020adversarial}: DNN-based malware detector hardened by {\em\underline{A}}dversarial {\em\underline{T}}raining with the ``{\em \underline{Max}}'' strategy enabled {\em\underline{M}}ixture of {\em\underline{A}}ttacks (AT-MaxMA);
    
    \item {\bf KDE} \cite{pang2018towards}: Combining DNN model with a secondary detector for quarantining adversarial examples. The detector is a {\em\underline{K}}ernel {\em\underline{D}}ensity {\em\underline{E}}stimator (KDE) built upon activations from the penultimate layer of DNN on normal examples;
    
    \item {\bf DLA}\cite{DBLP:conf/eurosp/SperlKCLB20}: The secondary detector aims to capture differences in DNN activations from the normal and adversarial examples. The adversarial examples are generated upon DNN. The activations from all dense layers are utilized, referred to as Dense Layer Analysis (DLA);
    
    \item {\bf DNN$^+$}\cite{DBLP:journals/corr/GrosseMP0M17,carlini2017adversarial}: The secondary detector plugs an extra class into the DNN model for detecting adversarial examples generated from DNN (DNN$^+$);
    
    \item {\bf ICNN}: The secondary detector is the {\em\underline{I}}nput {\em\underline{C}}onvexity {\em\underline{N}}eural {\em\underline{N}}etwork (ICNN), which is established upon the feature space and does not change the DNN (Section \ref{sec:ad-icnn});
    
    \item {\bf PAD-SMA}: {\em\underline{P}}rincipled {\em\underline{A}}dversarial {\em\underline{D}}etection is realized by a DNN-based malware detector and an ICNN-based adversary detector, both of which are hardened by adversarial training incorporating the {\em\underline{S}}tepwise {\em\underline{M}}ixture of {\em\underline{A}}ttacks (PAD-SMA, Algorithm \ref{alg:adv-train}).
\end{itemize}
At a high level, these defenses either harden the malware detector or introduce an adversary detector. 
%DNN serves as the baseline. 
More specifically, AT-rFGSM$^k$ can achieve better robustness than adversarial training methods with the BGA, BCA, or Grosse attack \cite{al2018adversarial}; AT-MaxMA with three PGD attacks can thwart a broad range of attacks but not iMaxMA, which is the iterative version of MaxMA \cite{li2020adversarial}; KDE, DLA, DNN$^+$ and ICNN aim to identify the adversarial examples by leveraging the underlying difference inherent in ML models between a pristine example and its variant; PAD-SMA hardens the combination of DNN and ICNN by adversarial training. 

% \smallskip
\noindent{\bf Metrics}. We report classification results on the test set via five standard metrics of False Negative Rate (FNR), False Positive Rate (FPR), F1 score, Accuracy (Acc for short, which is the percentage of the test examples that are correctly classified) and balanced Accuracy (bAcc) \cite{5597285}. Since we introduce $g$, a threshold $\tau$ is calculated on the validation set for rejecting examples. Let ``@\#'' denote the percentage of the examples in the validation set being outliers (e.g., @$5$ means 5\% of the examples are rejected by $g$).

\subsection{RQ1: Effectiveness in the Absence of Attacks} \label{sec:exp:dwoa}
\noindent{\bf Experimental setup}. We learn the aforementioned 8 detectors from the two datasets, respectively. In terms of malware detector model architecture, the DNN detector has 2 fully-connected hidden layers (each layer having 200 neurons) with ELU activation. The other 7 models also use this architecture. The adversary detector of DLA has the same setting as in \cite{DBLP:conf/eurosp/SperlKCLB20}: ICNN has 2 convex hidden layers with 200 neurons each. For adversarial training, feature representations can be flipped from ``0'' to ``1'' if injection operation is conducted and from ``1'' to ``0'' if removal operation is conducted. Moreover, AT-rFGSM$^k$ uses the PGD-$\ell_\infty$ attack, which additionally allows feature removals. It has 50 iterations with step size 0.02. AT-MaxMA uses three attacks, including PGD-$\ell_\infty$ iterates 50 times with step size 0.02, PGD-$\ell_2$ iterates 50 times with step size 0.5, and PGD-$\ell_1$ attack iterates 50 times, to conduct the training with penalty factor $\beta=0.01$ because a large $\beta$ incurs a low detection accuracy on the test sets. DLA and DNN$^+$ are learned from the adversarial examples generated by the MaxMA attack against the DNN model (i.e., adversarial training with an oblivious attack). 
PAD-SMA has three PGD attacks with the same step size as AT-MaxMA's except for $g$, which is learned from continuous perturbations. We set penalty factors $\beta_1=0.1$ and $\beta_2=1.0$ on the Drebin dataset and $\beta_1=0.01$ and $\beta_2=1.0$ on the Malscan dataset. In addition, we conduct a group of preliminary experiments to choose $\lambda$ from $\{10^{-3}, 10^{-2}, \ldots, 10^{3}\}$ and finally set $\lambda=1$ on both datasets. All detectors are tuned by the Adam optimizer with 50 epochs, mini-batch size 128, and learning rate 0.001, except for 80 epochs on the Malscan Dataset.

\begin{figure}[!htbp]
	\centering
	\begin{subfigure}[b]{0.35\textwidth}
		\centering
	\includegraphics[width=\textwidth]{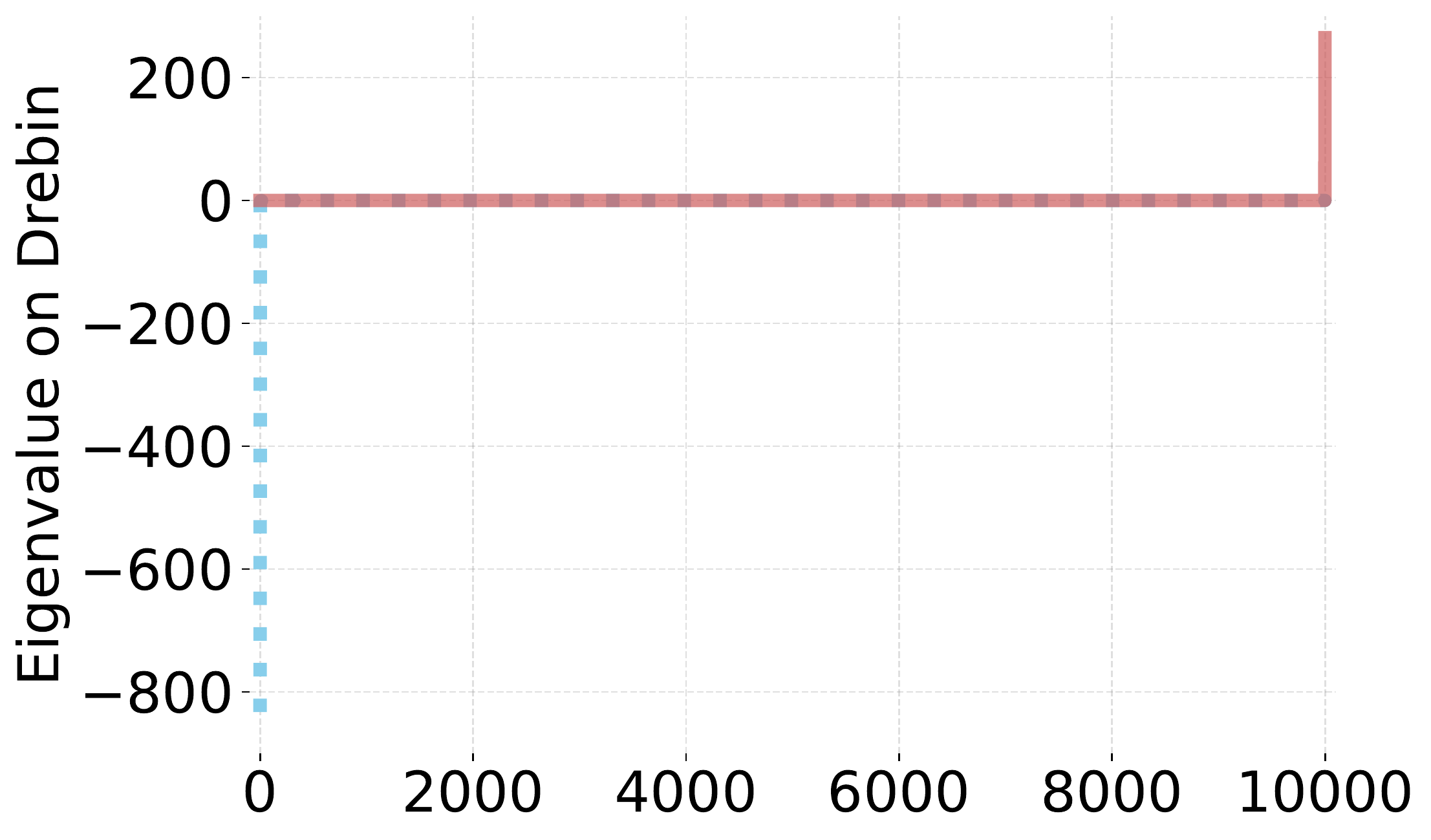}
		% \caption{Drebin}
		% \label{}
	\end{subfigure} \hspace{-9pt}
	\begin{subfigure}[b]{0.35\textwidth}
		\centering
	   % \vspace{-3pt}	 
        \includegraphics[width=\textwidth]{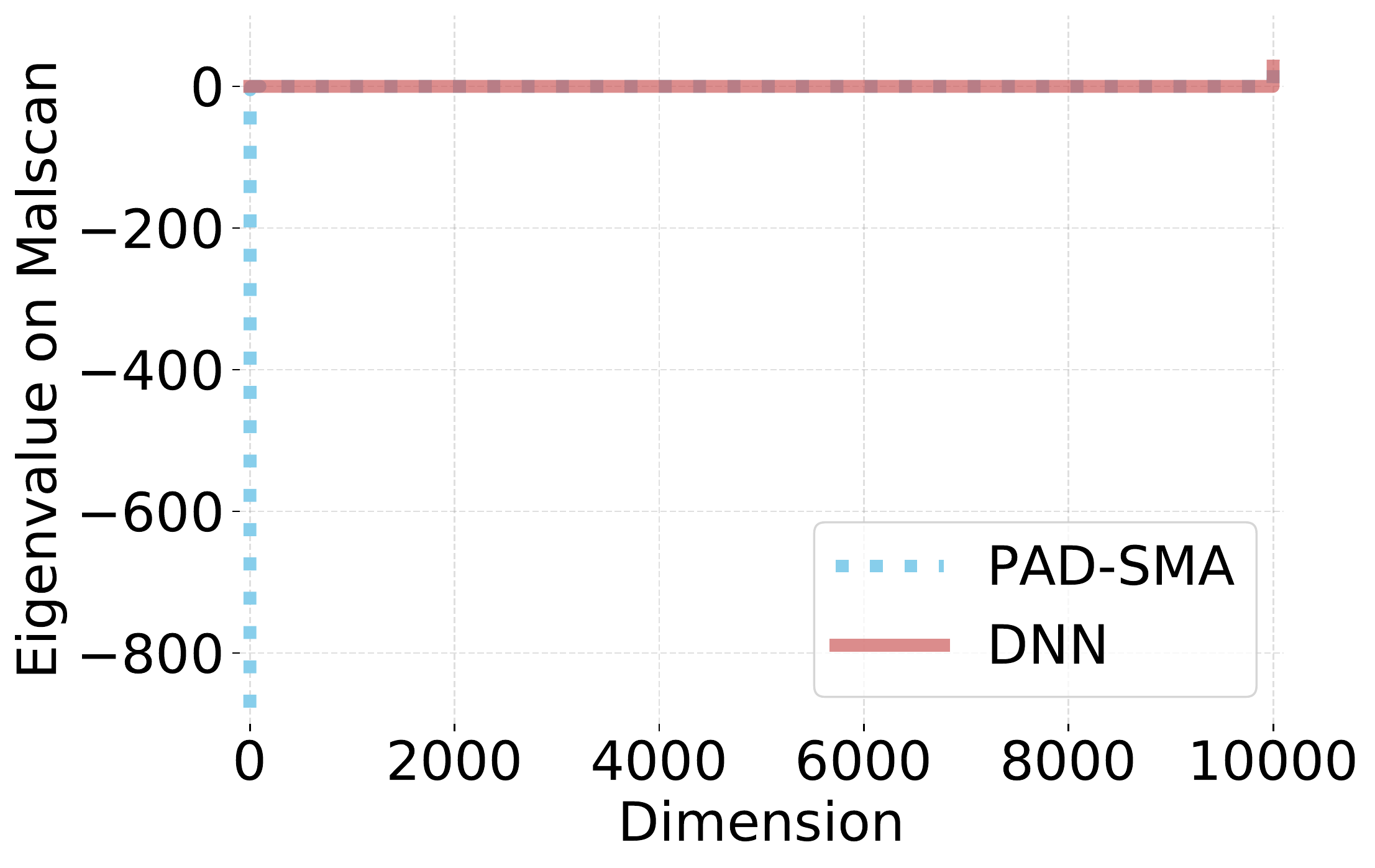}
		% \caption{Malscan}
		%\label{subfig-attack:02}
	\end{subfigure}
	\caption{Sorted eigenvalues of Hessian matrix of $\mathcal{F-\lambda\psi_\vartheta}$ w.r.t. input when $\lambda=1$. % {\color{red}Legend is missing in figure (a); when revising, using larger fonts as well}}
     % The number of eigenvalues corresponds to the number of input dimensions.
     }
    \label{fig:hessian}
    
\end{figure}

\noindent{\bf Experiments on confirming that PAD-SMA yields concave inner maximization}.
Figure \ref{fig:hessian} illustrates sorted eigenvalues of the Hessian matrix of the loss function $\ce-\psi_\vartheta$ w.r.t. input. We randomly choose 100 instance-label pairs from test datasets of Drebin and Malscan, respectively. We let these instances separately pass through PAD-SMA or DNN (which has $\psi_\vartheta=0$) for calculating eigenvalues, and then average the eigenvalues element-wisely corresponding to the input dimension. We observe that most eigenvalues are near $0$, PAD-SMA produces large negative eigenvalues, and DNN has relatively small positive eigenvalues. This shows that PAD-SMA can yield a concave inner maximization, confirming the theoretical results. Note that PAD-SMA still has positive eigenvalues on the Malcan dataset, and that robustness is achieved.
%as the existing adversarial training methods but there is no guarantee.}

\noindent{\bf Results answering RQ1}. Table \ref{tab:effe-no-attacks} reports the effectiveness of detectors on the two test sets. We observe that DNN achieves the highest detection accuracy (99.18\% on Drebin and 97.70\% on Malscan) and F1 score (96.45\% on Drebin and 97.73\% on Malscan). These accuracies are comparable to those reported in \cite{Daniel:NDSS,grosse2017adversarial,DBLP:conf/kbse/WuLZYZ019}. 
We also observe that KDE and ICNN have the same effectiveness as DNN because both are built upon DNN while introducing a separate model to detect adversarial examples. We further observe that when training with adversarial examples (e.g., AT-rFGSM$^k$, AT-MaxMA, DLA, DNN$^+$, and PAD-SMA), detectors' FNR decreases while FPR increases, leading to decreased F1 scores. This can be attributed to the fact that only the perturbed malware is used in the adversarial training and that data imbalance makes things worse.

\begin{table}[!t]
	\caption{
	Effectiveness (\%) of detectors without adversary detection capability in the absence of attacks.}
	\centering
	\begin{tabular}{c|l|ccccc}
		% \hline
		\toprule
		\multirow{2}{*}{} &\multirow{2}{*}{\specialcellleft{Defense}} & \multicolumn{5}{c}{Effectivenss (\%)} % \\\cline{3-7}
		\\\cmidrule{3-7}
		&&{FNR} & {FPR} & {Acc} & bAcc & {F1}\\
		% \hline
		\midrule
	 \multirow{8}{*}{\vthead{Drebin}} & \cc{DNN\cite{grosse2017adversarial}} & \cc{3.64} & \cc{0.45} & \cc{99.18} & \cc{97.96} & \cc{96.45}\\
		&AT-rFGSM$^k$\cite{al2018adversarial} & 2.36 & 3.43 & 96.69 & 97.10 & 87.18 \\
		&\cc{AT-MaxMA\cite{li2020adversarial}} &  \cc{1.73} & \cc{3.11} & \cc{97.05} & \cc{97.58} & \cc{88.46} \\
		&KDE\cite{pang2018towards} & 3.64 & 0.45 & 99.18 & 97.96 & 96.45 \\
        &\cc{DLA\cite{DBLP:conf/eurosp/SperlKCLB20}} & \cc{3.18} & \cc{0.58} & \cc{99.12} & \cc{98.12} & \cc{96.21} \\
		&DNN$^+$\cite{DBLP:journals/corr/GrosseMP0M17,carlini2017adversarial}& 3.36 & 0.50 & 99.17 & 98.07 & 96.42 \\
		& \cc{ICNN} & \cc{3.64} & \cc{0.45} & \cc{99.18} & \cc{97.96} & \cc{96.45} \\
            & PAD-SMA & 2.45 & 2.36 & 97.63 & 97.59 & 90.43 \\
		% \hline\hline
		\midrule\midrule
		\multirow{8}{*}{\vthead{Malscan}} & \cc{DNN\cite{grosse2017adversarial}} & \cc{1.87} & \cc{2.73} & \cc{97.70} & \cc{97.70} & \cc{97.73} \\
		&AT-rFGSM$^k$\cite{al2018adversarial} & 0.84 & 5.49 & 96.86 & 96.84 & 96.96 \\
		&\cc{AT-MaxMA\cite{li2020adversarial}} & \cc{0.39} & \cc{8.84} & \cc{95.43} & \cc{95.39} & \cc{95.65} \\
		&KDE\cite{pang2018towards} & 1.87 & 2.73 & 97.70 & 97.70 & 97.73 \\
		&\cc{DLA\cite{DBLP:conf/eurosp/SperlKCLB20}} & \cc{1.45} & \cc{3.35} & \cc{97.61} & \cc{97.60} & \cc{97.65} \\
		&DNN$^+$\cite{DBLP:journals/corr/GrosseMP0M17,carlini2017adversarial}& 2.81 & 1.84 & 97.67 & 97.68 & 97.68 \\
		& \cc{ICNN} & \cc{1.87} & \cc{2.73} & \cc{97.70} & \cc{97.70} & \cc{97.73} \\
            & PAD-SMA & 0.42 & 8.58 & 95.54 & 95.50 & 95.75 \\
		% \hline
		\bottomrule
	\end{tabular}
	\label{tab:effe-no-attacks}
\end{table}
\begin{table}[!t]
	\caption{
	Accuracy (\%) and F1 score (\%) of detectors with adversary detection capability in the absence of attacks.
	}
	\centering
	\begin{tabular}{c|l|cc|cc|cc}
		% \hline
		\toprule
		\multirow{2}{*}{} &\multirow{2}{*}{\specialcellleft{Defense}} & \multicolumn{2}{c|}{@1 (\%)} & \multicolumn{2}{c|}{@5 (\%)} &
        \multicolumn{2}{c}{@10 (\%)}% \\\cline{3-7}
		\\\cmidrule{3-8}
		&&{Acc} & {F1} & {Acc} & F1 & Acc & {F1}\\
		% \hline
		\midrule
		\multirow{5}{*}{\vthead{Drebin}} & \cc{KDE} & \cc{99.19} & \cc{96.45} & \cc{99.15} & \cc{96.33} & \cc{99.17} & \cc{96.43}\\
        & DLA & 99.14 &	96.27 & 99.13 & 96.27 & 99.14 & 96.53\\
		& \cc{DNN$^+$} & \cc{99.37} & \cc{97.20} & \cc{99.43} & \cc{97.44} & \cc{99.54} & \cc{97.93}\\
		& ICNN & 99.21 & 96.58 & 99.21 & 96.58 & 99.14 & 96.58\\
            & \cc{PAD-SMA} & \cc{97.79} & \cc{90.82} & \cc{97.99} & \cc{88.61} & \cc{98.14} &	\cc{79.54}\\
		% \hline\hline
		\midrule\midrule
		\multirow{5}{*}{\vthead{Malscan}} & KDE & 97.68 & 97.71 & 97.61 & 97.61 & 97.82 & 97.80\\
        & \cc{DLA} & \cc{97.65} & \cc{97.67} & \cc{97.69} & \cc{97.63} & \cc{97.80} & \cc{97.64}\\
		& DNN$^+$ & 97.81 & 97.81 & 98.37 & 98.38 & 98.58 & 98.56\\
		& \cc{ICNN} & \cc{97.68} & \cc{97.73} & \cc{97.64} & \cc{97.74} & \cc{97.70} & \cc{97.83}\\
            & PAD-SMA & 95.66 & 95.89 & 95.72 & 95.83 & 95.59 & 95.47 \\
		% \hline
		\bottomrule
	\end{tabular}
	\label{tab:effe-no-attacks2}
 \vspace{-3pt}
\end{table}

Table \ref{tab:effe-no-attacks2} reports the accuracy and F1 score of detectors with adversary detection capability $g$. To observe the behavior of $g$, we abstain $f$ from the prediction when $g(x)\geq\tau$. We expect to see that the trend of accuracy or F1 score will increase when removing as outliers more examples with high confidence from $g$ on the validation set. However, this phenomenon is not always observed (e.g., DLA and ICNN). This might be caused by the fact that DLA and ICNN distinguish the pristine examples confidently in the training phase, while the rejected examples on the validation set are in the distribution and thus have little impact on the detection accuracy of $f$. PAD-SMA gets the downtrend of F1 score but not accuracy, particularly on the Drebin dataset. Though this is counter-intuitive, we attribute it to the adversarial training with adaptive attacks, which implicitly pushes $g$ to predict the pristine malware examples with higher confidence than the benign ones. Thus, rejecting more validation examples actually causes more malware examples to be dropped, causing the remaining malware samples to be more similar to the benign ones and $f$ to misclassify remaining malware, leading to lower F1 scores. 

\begin{figure*}[!htbp]
	\centering
	\begin{subfigure}[b]{\textwidth}
		\centering
		\includegraphics[width=1\textwidth]{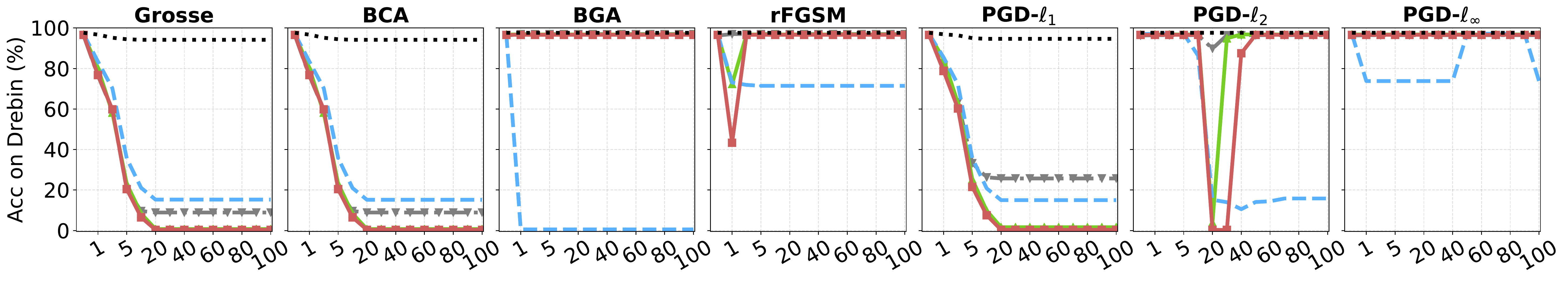}
		% \caption{}
		% \label{}
	\end{subfigure}
	
	\begin{subfigure}[b]{\textwidth}
		\centering
	\vspace{-3pt}	 
        \includegraphics[width=1\textwidth]{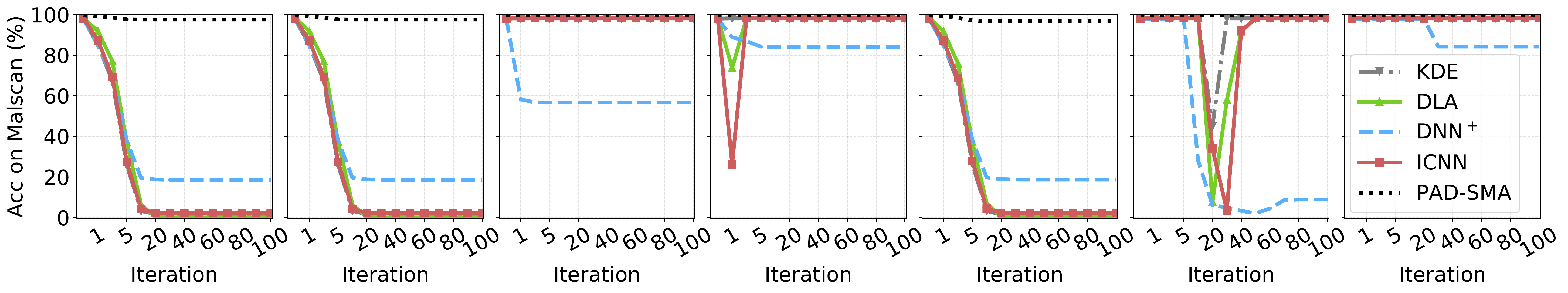}
		%\caption{}
		%\label{subfig-attack:02}
	\end{subfigure}
	\caption{Accuracy (Acc) of detectors against oblivious attacks with iteration from 0 to 100.}
	\label{fig:oblivion-attacks}
\end{figure*}

In summary, PAD-SMA decreases FNR but increases FPR, leading to decreased accuracies ($\leq$2.16\%) and F1 scores ($\leq$6.02\%), which aligns with the malware detectors learned from adversarial training. The use of adversary detectors in PAD-SMA does not make the situation better.

\begin{tcolorbox}[boxrule=0.5pt,arc=4pt,
left=1.5pt,right=1.5pt,top=3pt,bottom=3pt,boxsep=0pt]
\noindent{\bf Answer to RQ1}: There is no ``free lunch'' in the sense that using detectors trained from adversarial examples may suffer from a slightly lower accuracy when there are no adversarial attacks.
\end{tcolorbox}

\subsection{RQ2: Robustness against Oblivious Attacks} \label{sec:exp:dvo}

\noindent{\bf Experimental setup}. We measure the robustness of KDE, DLA, DNN$^+$, ICNN, and PAD-SMA against oblivious attacks via the Drebin and Malscan datasets; we do not consider the other detectors (i.e., DNN, AT-rFGSM$^k$, and AT-MaxMA) because they do not have $g$.
We use the detectors learned in the previous group of experiments (for answering RQ1). The threshold is computed by dropping 5\% validation examples with top confidence, which is suggested in \cite{pang2018towards,DBLP:conf/eurosp/SperlKCLB20,carlini2017adversarial}, while noting that the accuracy of PAD-SMA is slightly better than that of AT-MaxMT at this setting.

We separately wage 11 oblivious attacks to perturb malware examples on the test set. For Grosse \cite{grosse2017adversarial}, BCA \cite{al2018adversarial}, FGSM \cite{al2018adversarial}, BGA \cite{al2018adversarial}, PGD-$\ell_1$ \cite{DBLP:journals/tnse/LiLYX21}, PGD-$\ell_2$ \cite{DBLP:journals/tnse/LiLYX21}, and PGD-$\ell_\infty$\cite{DBLP:journals/tnse/LiLYX21}, these attacks proceed iteratively till the 100$^\text{th}$ loop is reached. Grosse, BCA, FGSM, and BGA are proposed to only permit the feature addition operation (i.e., flipping some `0's to `1's). FGSM has a step size 0.02 with random rounding. Three PGD attacks permit both feature addition and feature removal: PGD-$\ell_2$ has a step size 0.5 and PGD-$\ell_\infty$ has a step size 0.02 (the settings are the same as adversarial training). For Mimicry \cite{rndic_laskov}, we leverage $N_{ben}$ benign examples to guide the attack (dubbed Mimicry$\times N_{ben}$). We select the one that can evade $f$ to wage attacks and use a random one otherwise. MaxMA \cite{li2020adversarial} contains PGD-$\ell_1$, PGD-$\ell_2$, and PGD-$\ell_\infty$ attacks. The iterative MaxMA (dubbed iMaxMA) runs MaxMA 5 times, with the start point updated. SMA has 100 iterations with step size 0.5 for PGD-$\ell_2$ and 0.02 for PGD-$\ell_\infty$. The three MA attacks use the scoring rule of Eq.\eqref{eq:pgd-sel} without $g$ considered. 

\noindent{\bf Results}. Fig.\ref{fig:oblivion-attacks} depicts the accuracy curves of the detectors on Drebin (top panel) and Malscan (bottom panel) datasets under the 7 oblivious attacks, along with the iterations ranging from 0 to 100. We make three observations.
First, all these attacks cannot evade PAD-SMA (accuracy $\geq$ 90\%), demonstrating the robustness of the proposed model. 

Second, the Grosse, BCA, and PGD-$\ell_1$ attacks can evade KDE, DLA, DNN$^+$, and ICNN when 20 iterations are used, while recalling that these three attacks stop manipulating malware when the perturbed example can evade malware detector $f$. It is known that DNN is sensitive to small perturbations; KDE relies on the close distance between activations to reject large manipulations; DLA and DNN$^+$ are learned upon the oblivious MaxMA, which modifies malware examples to a large extent; ICNN is also learned from salt-and-pepper noises which randomly change one half elements of a vector. Therefore, neither malware detector $f$ nor adversary detector $g$ of KDE, DLA, and ICNN can impede small perturbations effectively. This explains why KDE, DLA, and ICNN can mitigate BGA and PGD-$\ell_\infty$ attacks that use large perturbations.

Third, a dip exists in the accuracy curve of KDE, DLA, or ICNN against rFGSM and PGD-$\ell_2$ when the iteration increases from 0 to 100. We find that both attacks can obtain small perturbations: rFGSM uses the random $\round$ (the rounding thresholds are randomly sampled from $[0,1]$) \cite{al2018adversarial} at iteration 1, and PGD-$\ell_2$ produces certain discrete perturbations at iteration 20 via $\round$ (the threshold is 0.5).

\begin{table}[!t]
	\caption{Accuracy (\%) of detectors under oblivious attacks (i.e., attacker is unaware of adversary detector $g$).}
		%, where KDE, GMM, and MAD are applied to enhancing DNN (rather than DNN+GAT in Table \ref{tab:ob-attack-res01}).}
	\centering
	\setlength{\tabcolsep}{5pt}
	\begin{tabular}{c|l|ccccc}
		% \hline
		\toprule
		\multirow{2}{*}{} &\multirow{2}{*}{Attack name}& \multicolumn{5}{c}{Accuracy (\%)} \\\cmidrule{3-7}
		&&{KDE} & {DLA} & DNN$^+$ & ICNN & PAD-SMA \\
		% \hline
		\midrule
		\multirow{7}{*}{\vthead{Drebin}}
		& \cc{No Attack} &  \cc{96.28} & \cc{96.80} & \cc{97.02} & \cc{96.62} & \cc{\bf 97.64}\\
            &Mimicry$\times 1$ & 56.64 & 55.82 & 58.18 & 54.91 & {\bf 94.18} \\
            &\cc{Mimicry$\times 10$} & \cc{20.91} & \cc{20.91} & \cc{23.55} & \cc{21.00} & \cc{\bf 84.18}\\
            &Mimicry$\times 30$ & 10.64 & 10.64 & 12.82 & 10.00 & {\bf 81.27} \\
		&\cc{MaxMA} & \cc{96.46} & \cc{96.82} & \cc{29.64} & \cc{96.64} & \cc{\bf 97.64} \\
            &iMaxMA & 96.46 & 96.82 & 29.64 & 96.64 & {\bf 97.64}\\
		&\cc{SMA} & \cc{32.09} & \cc{27.82} & \cc{31.18} & \cc{32.36} & \cc{\bf 94.27} \\
		% \hline\hline
		\midrule\midrule
		\multirow{7}{*}{\vthead{Malscan}} 
		& \cc{No Attack} & \cc{98.02} & \cc{98.41} & \cc{97.86} & \cc{98.11} & \cc{\bf 99.65} \\
            &Mimicry$\times 1$ & 49.74 & 53.65 & 47.81 & 49.32 & {\bf 83.68} \\
            &\cc{Mimicry$\times 10$} & \cc{18.13} & \cc{18.68} & \cc{21.68} & \cc{17.06} & \cc{\bf 69.13} \\
            &Mimicry$\times 30$ & 8.65 & 6.94 & 14.23 & 7.00 & {\bf 65.45}\\
		&\cc{MaxMA} & \cc{98.13} & \cc{98.55} & \cc{84.23} & \cc{98.16} & \cc{\bf 99.65} \\
            &iMaxMA & 98.13 & 98.55 & 84.23 & 98.16 & {\bf 99.65} \\
		&\cc{SMA} & \cc{6.00} & \cc{26.68} & \cc{19.03} & \cc{7.32} & \cc{\bf 96.68}\\
		% \hline
		\bottomrule
	\end{tabular}
	\label{tab:ob-attack-res02}
\end{table}

\begin{table*}[!htbp]
	\caption{Accuracy (\%) of detectors under adaptive attacks, where ``Orth'' stands for ``orthogonal'', ``$-$'' means an attack is not applicable.
    }
	\centering
	\begin{tabular}{c|l|cccccccc}
		% \hline
		\toprule
		\multirow{2}{*}{} &\multirow{2}{*}{Attack name}& \multicolumn{8}{c}{Accuracy (\%)} \\\cmidrule{3-10}
		&& DNN & AT-rFGSM & AT-MaxMA & {KDE} & {DLA} & DNN$^+$ & ICNN & PAD-SMA \\
		% \hline
		\midrule
		\multirow{17}{*}{Drebin}
		& \cc{Groose} & \cc{0.000} & \cc{48.00} & \cc{ 87.64} & \cc{0.000} & \cc{0.000} & \cc{0.000} & \cc{0.636} & \cc{{\bf 90.91}} \\
            & BCA & 0.000 & 47.73 & 87.64 & 6.182 & 0.000 & 4.727 & 3.000 & {\bf 93.00} \\
            & \cc{BGA} & \cc{0.000} & \cc{95.55} & \cc{96.64} & \cc{97.00} & \cc{2.455} & \cc{0.000} & \cc{33.36} & \cc{\bf 97.64} \\
            & rFGSM & 0.000 & 97.46 & {\bf 98.18} & 97.00 & 96.82 & 70.91 & 96.64 & {97.64} \\
		& \cc{PGD-$\ell_1$} & \cc{0.000} & \cc{44.46} & \cc{80.91} & \cc{0.182} & \cc{0.000} & \cc{0.000} & \cc{0.091} & \cc{\bf 89.72} \\
            & PGD-$\ell_2$ & 3.455 & 89.73 & 96.27 & 87.36 & 0.000 & 8.727 & 0.091 & {\bf 97.18} \\
            & \cc{PGD-$\ell_\infty$} & \cc{0.000} & \cc{96.55} & \cc{\bf 98.09} & \cc{97.00} & \cc{96.82} & \cc{63.73} & \cc{96.64} & \cc{97.46} \\
            & Mimicry$\times 1$ & 54.91 & 88.91 & 90.27 & 56.64 & 55.82 & 58.18 & 54.91 & {\bf 94.18} \\
            & \cc{Mimicry$\times 10$} & \cc{21.00} & \cc{71.82} & \cc{74.27} & \cc{25.73} & \cc{20.36} & \cc{19.18} & \cc{21.00} & \cc{\bf 81.18} \\
            & Mimicry$\times 30$ & 10.00 & 66.45 & 70.64 & 16.09 & 10.09 & 7.909 & 10.00 & {\bf 74.27} \\
            & \cc{MaxMA} & \cc{0.000} & \cc{44.36} & \cc{80.64} & \cc{0.182} & \cc{0.000} & \cc{0.000} & \cc{0.091} & \cc{\bf 89.09} \\
            & iMaxMA & 0.000 & 43.36 & 69.64 & 0.000 & 0.000 & 0.000 & 0.000 & {\bf 88.73} \\
		& \cc{SMA} & \cc{0.000} & \cc{57.82} & \cc{84.09} & \cc{16.36} & \cc{0.000} & \cc{8.636} & \cc{0.000} & \cc{\bf 94.46} \\
            & Orth PGD-$\ell_1$ & $-$ & $-$ & $-$ & 1.091 & 0.000 & 0.000 & 0.000 & {\bf 97.64} \\
            & \cc{Orth PGD-$\ell_2$} & \cc{$-$} & \cc{$-$} & \cc{$-$} & \cc{17.46} & \cc{2.455} & \cc{13.55} & \cc{3.909} & \cc{\bf 97.64} \\
            & Orth PGD-$\ell_\infty$ & $-$ & $-$ & $-$ & 96.82 & 31.73 & 55.18 & 96.46 & {\bf 97.64} \\
            & \cc{Orth MaxMa} & \cc{$-$} & \cc{$-$} & \cc{$-$} & \cc{1.091} & \cc{0.000} & \cc{0.000} & \cc{0.000} & \cc{\bf 97.64} \\
            & Orth iMaxMa & $-$ & $-$ & $-$ & 0.182 & 0.000 & 0.000 & 0.000 & {\bf 97.64}\\
		% \hline\hline
		\midrule\midrule
		\multirow{17}{*}{Malscan} 
		& \cc{Groose} & \cc{0.000} & \cc{9.129} & \cc{77.26} & \cc{0.000} & \cc{0.000} & \cc{0.000} & \cc{0.871} & \cc{\bf 85.26} \\
            & BCA & 0.000 & 8.968 & 77.03 & 1.194 & 0.000 & 0.097 & 8.129 & {\bf 89.32} \\
            & \cc{BGA} & \cc{0.000} & \cc{10.97} & \cc{95.68} & \cc{98.13} & \cc{0.194} & \cc{30.19} & \cc{37.45} & \cc{\bf 99.45}\\
            & rFGSM & 0.000 & 99.16 & 99.55 & 98.13 & 98.55 & 83.42 & 98.16 & {\bf 99.65}\\
		  & \cc{PGD-$\ell_1$} & \cc{0.000} & \cc{6.000} & \cc{71.68} & \cc{0.000} & \cc{0.000} & \cc{0.000} & \cc{1.226} & \cc{\bf 84.87} \\
            & PGD-$\ell_2$ & 34.13 & 63.94 & 81.55 & 38.32 & 2.097 & 2.806 & 2.548 & {\bf 95.90} \\
            & \cc{PGD-$\ell_\infty$} & \cc{0.000} & \cc{99.16} & \cc{\bf 99.52} & \cc{98.13} & \cc{98.55} & \cc{41.07} & \cc{98.10} & \cc{99.45} \\
            & Mimicry$\times 1$ & 49.32 & 75.39 & 82.48 & 49.74 & 53.65 & 47.81 & 49.32 & {\bf 83.68} \\
            & \cc{Mimicry$\times 10$} & \cc{17.06} & \cc{49.13} & \cc{60.71} & \cc{17.52} & \cc{18.23} & \cc{11.65} & \cc{17.06} & \cc{\bf 59.94} \\
            & Mimicry$\times 30$ & 7.000 & 39.94 & 52.48 & 7.645 & 6.483 & 2.452 & 7.000 & {\bf 53.68} \\
            & \cc{MaxMA} & \cc{0.000} & \cc{5.742} & \cc{61.77} & \cc{0.645} & \cc{0.000} & \cc{0.000} & \cc{0.935} & \cc{\bf 85.26} \\
            & iMaxMA & 0.000 & 1.645 & 47.07 & 0.097 & 0.000 & 0.000 & 0.935 & {\bf 83.45} \\
		  & \cc{SMA} & \cc{0.000}  & \cc{28.77} & \cc{78.36} & \cc{0.323} & \cc{8.258} & \cc{1.000} & \cc{0.903} & \cc{\bf 97.48} \\
            & Orth PGD-$\ell_1$  & $-$ & $-$ & $-$ & 2.000 & 0.000 & 0.032 & 0.000 & {\bf 99.65} \\
            & \cc{Orth PGD-$\ell_2$}  & \cc{$-$} & \cc{$-$} & \cc{$-$} & \cc{38.32} & \cc{2.097} & \cc{2.806} & \cc{2.548} & \cc{\bf 99.65}\\
            & Orth PGD-$\ell_\infty$ & $-$ & $-$ & $-$ & 98.13 & 87.97 & 34.23 & 98.16 & {\bf 99.65} \\
            & \cc{Orth MaxMa} & \cc{$-$} & \cc{$-$} & \cc{$-$} & \cc{1.806} & \cc{0.000} & \cc{0.032} & \cc{0.000} & \cc{\bf 99.65} \\
            & Orth iMaxMa & $-$ & $-$ & $-$ & 0.484 & 0.000 & 0.032 & 0.000 & {\bf 99.65}\\
		% \hline
		\bottomrule
	\end{tabular}
	\label{tab:wb-attack-res}
\end{table*}

Table \ref{tab:ob-attack-res02} reports the attack results of Mimicry, MaxMA, iMaxMA, and SMA, which are not suitable for iterating with a large number of loops. We make three observations. First, PAD-SMA can effectively defend against these attacks, except for Mimicry$\times 30$ (with an accuracy of 65.45\% on Malscan). Mimicry attempts to modify malware representations to resemble benign ones. As reported in Section \ref{sec:exp:dwoa}, adversarial training promotes ICNN ($g$ of PAD-SMA) to implicitly distinguish malicious examples from benign ones. Both aspects decrease PAD-SMA's capability in mitigating the oblivious Mimicry attack effectively. 
Second, all detectors can resist MaxMA and iMaxMA, except for DNN$^+$. Both attacks maximize the classification loss of DNN$^+$, leading DNN$^+$ to misclassify perturbed examples as benign (rather than the newly introduced label). Third, all detectors are vulnerable to the SMA attack (with maximum accuracy of 32.36\% on Drebin and 26.68\% on Malscan), except for PAD-SMA. This is because SMA stops perturbing malware when a successful adversarial example against $f$ is obtained although the degree of perturbations is small, which cannot be identified by $g$ of KDE, DLA, DNN$^+$, or ICNN.

\begin{tcolorbox}[boxrule=0.5pt,arc=4pt,
left=1.5pt,right=1.5pt,top=3pt,bottom=3pt,boxsep=0pt]
\noindent{\bf Answer to RQ2}: PAD-SMA is significantly more robust than KDE, DLA, DNN$^+$, and ICNN against oblivious attacks. Still, PAD-SMA cannot effectively resist the Mimicry attacks that are guided by multiple benign samples. 
\end{tcolorbox}

\subsection{RQ3: Robustness against Adaptive Attacks} \label{sec:exp:dva}

\noindent{\bf Experimental setup}. We measure the robustness of the detectors against adaptive attacks on the Drebin and Malscan datasets. We use the 8 detectors in the first group of experiments. The threshold $\tau$ is set as the one in the second group of experiments unless explicitly stated otherwise.
The attacker knows $f$ and $g$ (if applicable) to manipulate malware examples on the test sets. We change the 11 oblivious attacks to adaptive attacks by using the loss function given in Eq.\eqref{eq:lag}, which contains both $\ce$ and $\psi_\vartheta$. When perturbing an example, a linear search is conducted to look for a $\lambda$ from the set of $\{10^{-5},\ldots,10^{5}\}$. In addition, the Mimicry attack can query both $f$ and $g$ and get feedback then. On the other hand, since DNN, AT-rFGSM, and AT-MaxMA contain no adversary detector, the oblivious attacks trivially meet the adaptive requirement. The other 5 attacks are adapted from orthogonal (Orth for short) PGD \cite{DBLP:journals/corr/abs-2106-15023}, including Orth PGD-$\ell_1$, PGD-$\ell_2$, PGD-$\ell_\infty$, MaxMA, and iMaxMA. We use the scoring rule of Eq.\eqref{eq:pgd-sel} to select the orthogonal manner. The hyper-parameters of attacks are the same as the second group of experiments, except for PGD-$\ell_1$ using 500 iterations, PGD-$\ell_2$ using 200 iterations with step size 0.05, and PGD-$\ell_\infty$ using 500 iterations with step size 0.002.

\noindent{\bf Results}. Table \ref{tab:wb-attack-res} summarizes the experimental results. We make three observations. First, DNN is vulnerable to all attacks %totally ineffective against 9 attacks (
with 0\% accuracy. The Mimicry attack achieves the lowest effectiveness in evading DNN because it modifies examples without using the internal information of victim detectors. AT-rFGSM can harden the robustness of DNN to some extent, but is still sensitive to BCA, PGD-$\ell_1$, MaxMa, and iMaxMA attacks (with an accuracy $\leq$ 47.73\% on both datasets). With an adversary detector, KDE, DLA, DNN$^+$, and ICNN can resist a few attacks (e.g., rFGSM and PGD-$\ell_\infty$), but the effectiveness is limited. AT-MaxMA impedes a range of attacks except for iMaxMA (with a $69.94\%$ accuracy on Drebin and $47.07\%$ on Malscan) and Mimicry$\times 30$ (with a $70.64\%$ accuracy on Drebin and $52.48\%$ on Malscan), which are consistent with previous results \cite{li2020adversarial}.

Second, PAD-SMA significantly outperforms the other defenses (e.g., AT-MaxMA), by achieving robustness against 16 attacks on the Drebin dataset and 13 attacks on the Malscan dataset (with accuracy $\geq 85\%$). For example, PAD-SMA can mitigate MaxMA and iMaxMA, while AT-MaxMA can resist MaxMA but not iMaxMA (accuracy dropping by 11\% on Drebin and 14.7\% on Malscan). The reason is that PAD-SMA is optimized with convergence guaranteed, causing that more iterations do not promote attack effectiveness, which resonates our theoretical results. Moreover, PAD-SMA gains high detection accuracy ($\geq 97.64\%$) against orthogonal attacks because the same scoring rule is used and PAD-SMA renders loss function concave. 

Third, Mimicry$\times 30$ can evade all defenses (with accuracy $\leq$ 74.27\% on Drebin and $\leq$ 53.68\% on Malscan). We additionally conduct two experiments on Drebin: (i) when we retrain PAD-SMA with penalty factor $\beta_1$ increased from $\beta_1=0.1$ to $\beta_1=1.0$, the detection accuracy increases to 85.27\% against Mimicry$\times 30$ with the detection accuracy on the test dataset decreasing notably (F1 score decreasing to 78.06\%); (ii) when we train PAD-SMA on Mimicry$\times 30$ with additional 10 epochs, the robustness increases to 83.64\% against Mimicry$\times 30$ but  the detection accuracy also decreases on the test set. These hint that our method, as other adversarial malware training methods, suffers from a trade-off between robustness and accuracy.

\begin{tcolorbox}[boxrule=0.5pt,arc=4pt,
left=1.5pt,right=1.5pt,top=3pt,bottom=3pt,boxsep=0pt]
\noindent{\bf Answer to RQ3}: PAD-SMA outperforms the other defenses, by significantly hardening malware detectors against a range of adaptive attacks but not Mimicry$\times 30$. 
\end{tcolorbox}

\subsection{RQ4: Robustness against Practical Attacks} \label{sec:exp:dvp}

\noindent{\bf Experimental setup}. We implement a system to produce adversarial malware for all attacks considered. We handle the inverse feature mapping problem (Section \ref{sec:ata}) as in \cite{li2020adversarial}, by mapping perturbations in the feature space to the problem space. Our manipulation proceeds as follows: (i) obtain feature perturbations; (ii) disassemble an app using Apktool \cite{apktool:Online}; (iii) perform manipulation and assemble perturbed files using Apktool. We add manifest features and do not remove them for preserving an app's functionality. We permit all APIs that can be added and the APIs with {\tt public} modifier but no {\em class inheritance} can be hidden by the reflection technique 
(see supplementary materials for details). In addition, the functionality estimation is conducted by {\em Android Monkey}, which is an efficient fuzz testing tool that can randomly generate app activities to execute on Android devices, along with logs. If an app and its modified version have the same activities, we treat them as having the same functionality. However, we manually re-analyze the non-functional ones to cope with the randomness of Monkey. We wage Mimcry$\times30$, iMaxMA, and SMA attacks because they achieve a high evasion capability in the feature space. 

\noindent{\bf Results}. We respectively modify 1,098, 1,098, and 1,098 apps by waging the Mimcry$\times30$, iMaxMA, and SMA attacks to the Drebin test set (leading to 1,100 malicious apps in total), and 2,790, 2,791, and 2,790 apps to the Malscan test set (leading to 3,100 malicious apps in total). Most failed cases are packed apps against ApkTool.

\begin{table}[!htbp]
	\caption{The number of apps with functionalities preserved from 100 randomly selected examples.}
	\centering
	\setlength{\tabcolsep}{3.5pt}
	\begin{tabular}{c|c|cccc}
		% \hline
		\toprule
		\multirow{2}{*}{Dataset} &\multirow{2}{*}{\specialcell{Functionality}}& \multicolumn{4}{c}{Apps (\#)} \\\cmidrule{3-6}
		&&No attack & Mimicry$\times30$ & iMaxMA & SMA \\
		% \hline
		\midrule
		\multirow{2}{*}{Drebin}& \cc{Installable} & \cc{89} & \cc{89} & \cc{89} & \cc{89} \\
		& Monkey & 80 & 68 & 66 & 65 \\
		% \hline\hline
		\midrule\midrule
		\multirow{2}{*}{\specialcell{Andro-\\zoo}}& \cc{Installable} & \cc{86} & \cc{84} & \cc{86} & \cc{83} \\
		& Monkey & 76 & 58 & 65 & 64 \\
		
		% \hline
		\bottomrule
	\end{tabular}
	\label{tab:func-est}
\end{table}

Table \ref{tab:func-est} reports the number of modified apps that retain the malicious functionality. Given 100 randomly chosen apps, 89 apps on Drebin and 86 apps on Malscan can be deployed on an Android emulator (running Android API version 8.0 and ARM library supported). Monkey testing says that the ratio of functionality preservation is at least 73.03\% (65 out of 89) on the Drebin dataset and 69.05\% (58 out of 84) on the Malscan dataset. Through manual inspection, we find that the injection of {\tt null} constructor cannot pass the verification mechanism of the Android Runtime. Moreover, Java reflection sometimes breaks an app's functionality when the app verifies whether an API name is changed and then chooses to throw an error.

\begin{figure}[!htbp]
	\centering
	\includegraphics[width=0.5\textwidth]{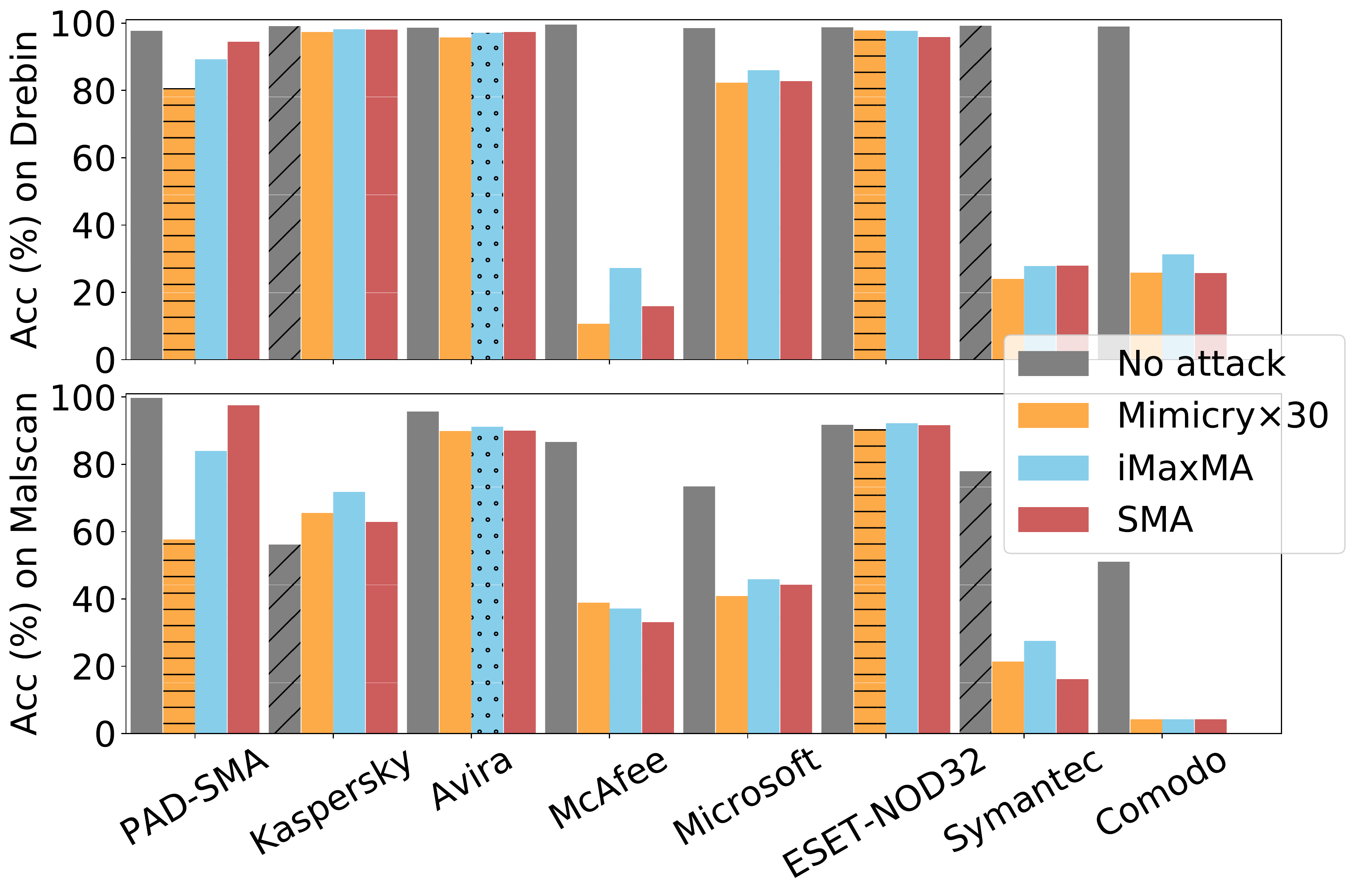}
	\caption{Effectiveness of PAD-SMA and malware scanners against practical attacks.}
	\label{fig:anti-scanners}
\end{figure}

Fig.\ref{fig:anti-scanners} depicts the detection accuracy of detectors against Mimicry$\times 30$, iMaxMA, and SMA attacks. We observe that PAD-SMA cannot surpass Avira and ESET-NOD32 on both the Drebin and Malscan datasets. Note that these attacks know the feature space of PAD-SMA but not anti-malware scanners. Nevertheless, PAD-SMA achieves comparable robustness to the three attacks by comparing with Microsoft, and outperforms McAfee, Symantec, and Comodo. In addition, Kaspersky is seemingly adaptive to these attacks because it obtains a slightly better accuracy on the modified apps than the unperturbed ones ($\leq$15.59\%) on the Malscan dataset.

\begin{tcolorbox}[boxrule=0.5pt,arc=4pt,
left=1.5pt,right=1.5pt,top=3pt,bottom=3pt,boxsep=0pt]
\noindent{\bf Answer to RQ4}: PAD-SMA is comparable to anti-malware scanners in the presence of practical attacks. It effectively mitigate iMaxMA and SMA attacks, but has limited success against Mimicry$\times 30$, akin to the cases of circumventing feature-space attacks.
\end{tcolorbox}
\section{Related Work} 
\label{sec:related-work}

We divide related prior studies into two classes: Adversarial Malware Detection (AMD) vs. Adversarial ML (AML).

\noindent{\bf Defenses against adversarial examples in AMD}. We further divide the related literature into three categories: (i) robust feature extraction, (ii) learning model enhancement, and (iii) adversarial example detection. 

In terms of robust feature extraction, Drebin features, including manifest instructions (e.g., required permissions) and syntax instructions (e.g., sensitive APIs), are usually applied to resist adversarial examples \cite{Daniel:NDSS,grosse2017adversarial,suciu2018does,pierazzi2019intriguing}. Furthermore, Demontis et al. \cite{7917369} demonstrate the robustness of Drebin features using several evasion attacks. However, a following study questions this observation with a mixture of attacks \cite{li2020adversarial}. Moreover, to cope with {\em obfuscation} attacks, researchers suggest leveraging system API calls \cite{hou2017hindroid}, and further enrich the representation by incorporating multiple modalities such as structural information (e.g., call graph), API usage (e.g., method argument types, API dependencies), and dynamic behaviors (e.g., network activity, memory dump) \cite{onwuzurike2019mamadroid,10.1145/3372297.3417291,DBLP:journals/tc/Ficco22}. In this paper, we mainly focus on improving the robustness of the learning model, although the feature robustness is also important. Therefore, we refine Drebin features by filtering the ones that can be easily manipulated.

In terms of learning model enhancement,
%For the second aspect, 
the defense mechanisms aim to enhance a malware detector itself to classify adversarial examples accurately. Several approaches exist, such as classifier randomization, ensemble learning, input transformation, and adversarial training, which are summarized by a recent survey \cite{DBLP:journals/corr/abs-2005-11671}. We focus on adversarial training, which augments the training dataset with adversarial examples \cite{xu2014evasion,DBLP:conf/eisic/ChenYB17,grosse2017adversarial,al2018adversarial}. In order to promote the robustness, the min-max adversarial training \cite{madry2017towards} in machine learning is adapted to the context of malware detection, aiming to make detectors perceive the optimal attack in a sense to resist non-optimal ones \cite{al2018adversarial,DBLP:journals/tnse/LiLYX21}. In practice, the attackers are free enough to generate multiple types of adversarial examples, straightly leading to the instantiation of adversarial training incorporating a mixture of attacks \cite{li2020adversarial}. In addition, combining adversarial training and ensemble learning further promotes robustness as long as the base model has a due amount of robustness \cite{li2020adversarial}; a recent study demonstrates that diversified features also promote the robustness of ensemble model \cite{DBLP:journals/tc/Ficco22}.
This paper aims to establish principled min-max adversarial training methods with rigorous robustness. Moreover, a new mixture of attacks is used to instantiate our framework.

In terms of adversarial example detection, the defenses aim to identify adversarial examples for further analysis. There are two approaches. The first approach is to study detectors based on traditional ML models such as ensemble learning based (e.g., \cite{smutz2016tree}). Inspired by the observation that grey-box attacks cannot thwart all basic building-block classifiers, Smutz et al. \cite{smutz2016tree} propose identifying evasion attacks via prediction confidences. However, it is not clear how to adapt these ideas to deep learning models because they leverage properties which may not exist in DL models (e.g., neural networks are poorly, rather than well, calibrated \cite{guo2017calibration}). The second approach is to leverage the invariant in malware features or detectors to recognize adversarial examples. For example, Grosse et al. \cite{DBLP:journals/corr/GrosseMP0M17} demonstrate the difference between examples and their perturbed versions using statistical tests. Li et al. \cite{li2018hashtran} and Li et al. \cite{DBLP:conf/www/LiZYLGC21} respectively propose detecting adversarial examples via stacked denoising autoencoders. However, these defense models seemingly cannot deal with adaptive attacks \cite{DBLP:journals/corr/GrosseMP0M17,li2018hashtran,DBLP:journals/corr/abs-2106-15023}. Moreover, some defense models are not validated with adaptive attacks \cite{DBLP:conf/www/LiZYLGC21}. When compared with these prior studies, our solution leverages a convex DNN model to recognize the evasion attacks, which is not only able to detect adversarial examples, but also able to promote principled defenses \cite{DBLP:conf/iclr/SinhaND18}, leading to a formal treatment on robustness. Although our model has malware and adversary detectors, it is different from ensemble learning because they use different losses.

\noindent{\bf Adversarial training in AML}. 
%We review related literature on adversarial training. 
Adversarial training augments the training set with adversarial examples \cite{szegedyZSBEGF13,goodfellow6572explaining}. Multiple heuristic strategies have been proposed to generate adversarial examples, including the one that casts adversarial training as a min-max optimization problem \cite{madry2017towards}. It minimizes the loss for learning ML models upon the most powerful attack (i.e., considering the worst-case scenario). However, owing to the non-linearity of DNNs, it is NP-hard to solve the inner maximization exactly \cite{madry2017towards}. There are two lines of studies to improve the min-max adversarial training: one aims to select or produce the optimal adversarial examples (e.g., via advanced criterion or new learning strategies  \cite{DBLP:conf/iclr/TramerKPGBM18,DBLP:conf/iclr/0001ZY0MG20,DBLP:conf/ijcai/BaiL0WW21,DBLP:conf/cvpr/JiaZW0WC22}); the other aims to analyze statistical properties of resulting models (e.g.,via specific NN architectures or convexity assumptions \cite{DBLP:conf/aistats/XingSC21a,DBLP:conf/iclr/SinhaND18}). However, adversarial training is domain-specific, meaning that it is non-trivial to leverage these advancements for enhancing ML-based malware detectors.

\section{Conclusion} \label{sec:conclusion}

We devised a provable defense framework  for malware detection against adversarial examples. Instead of hardening the malware detector solely, we use an indicator to alert the presence of adversarial examples. We instantiate the framework via adversarial training with a new mixture of attacks, along with a theoretical analysis on the resulting robustness. Experiments with two Android datasets demonstrate the soundness of the framework against a set of attacks, including 3 practical ones.  
Future research needs to design other principled or verifiable methods. Learning or devising robust features, especially dynamic analysis based features, may be key to detecting adversarial examples.
Other open problems include unifying practical adversarial malware attacks, designing application-agnostic manipulations, and formally verifying functionality-preservation and model robustness.

\ifCLASSOPTIONcaptionsoff
  \newpage
\fi

\bibliographystyle{IEEEtran} % IEEEtranS
\bibliography{IEEEabrv,malware}

\newpage
\clearpage
\appendices
\section{Theorem Proofs}
\subsection{Notations}
Table \ref{table:notations} summarizes the notations 
%that are used throughout the paper 
for improving the readability of the proofs.
%supplementary material.

\subsection{Proposition \ref{prop:framework}}
\begin{proposition*}
	Given  continuous function $\ce$, and continuous and convex distance $C(\cdot, \mathbf{x})=\max\{0, \psi_\vartheta(\cdot)-\tau\}$ with $\mathbf{x}\sim\mathbb{P}$, the dual problem of $\max\limits_{\mathbb{P}':W(\mathbb{P}',\mathbb{P})\leq0}\mathbb{E}_{\mathbf{x}'\sim\mathbb{P}'}\ce(\theta,\mathbf{x}',1)$ is 
	$$\inf_{\lambda}\Bigl\{\mathbb{E}_{\mathbf{x}\sim\mathbb{P}}\max\limits_{\delta_\mathbf{x}}(\ce(\theta,\mathbf{x}+\delta_\mathbf{x},1)-\lambda\psi_\vartheta(\mathbf{x}+\delta_\mathbf{x}) + \lambda\tau):\lambda\geq 0\Bigl\},$$ where $\mathbf{x}+\delta_\mathbf{x}\in\mathcal{X}$, $\psi_\vartheta(\mathbf{x}+\delta_\mathbf{x})\geq \tau$ and $W(\mathbb{P}',\mathbb{P}):=\inf\limits_{\Gamma}\left\{\int C(\mathbf{x}',\mathbf{x})d\Gamma(\mathbf{x}',\mathbf{x}):\Gamma\in\prod(\mathbb{P}',\mathbb{P})\right\}$.
\end{proposition*}
\begin{proof}
The proof is adapted from the one presented in \cite{DBLP:conf/iclr/SinhaND18}.
	\begin{fleqn}
	\begin{align*}
	&\max\limits_{\mathbb{P}':W(\mathbb{P}',\mathbb{P})\leq0}\mathbb{E}_{\mathbf{x}'\sim\mathbb{P}'}\ce(\theta,\mathbf{x}',1) \\
		=&\max\limits_{\mathbb{P}':W(\mathbb{P}',\mathbb{P})\leq0}\inf_{\lambda\geq 0}\left\{\mathbb{E}_{\mathbf{x}'\sim\mathbb{P}'}\left[\ce(\theta,\mathbf{x}',1)\right]-\lambda W(\mathbb{P}',\mathbb{P})\right\}\\
		\overset{\footnotesize \whitecircled{1}}{=}&\inf_{\lambda\geq 0}\max\limits_{\mathbb{P}':W(\mathbb{P}',\mathbb{P})\leq0}\left\{\mathbb{E}_{\mathbf{x}'\sim\mathbb{P}'}\left[\ce(\theta,\mathbf{x}',1)\right]-\lambda W(\mathbb{P}',\mathbb{P})\right\} \\
		=&\inf_{\lambda\geq 0}\max\limits_{\Gamma:W(\mathbb{P}',\mathbb{P})\leq0}\left\{\mathbb{E}_{(\mathbf{x}',\mathbf{x})\sim\Gamma}\left[\ce(\theta,\mathbf{x}',1)-\lambda C(\mathbf{x}',\mathbf{x})\right]\right\} \\
		\leq&\inf_{\lambda\geq 0}\left\{\mathbb{E}_{\mathbf{x}\sim\mathbb{P}}\left[\max\limits_{\mathbf{x}'}\left(\ce(\theta,\mathbf{x}',1)-\lambda C(\mathbf{x}',\mathbf{x})\right)\right]\right\},
	\end{align*}
	\end{fleqn}
where ${\footnotesize \whitecircled{1}}$ holds because of Slater’s condition. Recall that $\mathbf{x}'$ is perturbed from $\mathbf{x}$, this constraint leads to 
	\begin{fleqn}
	\begin{align*}
		&\max\limits_{W(\mathbb{P}',\mathbb{P})\leq0}\left\{\mathbb{E}_{(\mathbf{x}',\mathbf{x})\sim\Gamma}\left[\ce(\theta,\mathbf{x}',1)-\lambda C(\mathbf{x}',\mathbf{x})\right]\right\} \\
		\geq & \mathbb{E}_{\mathbf{x}\sim\mathbb{P}}\left\{\max\limits_{\mathbb{P}':W(\mathbb{P}',\mathbb{P})\leq0}\left[\mathbb{E}_{\mathbf{x}'\sim \mathbb{P}'|\mathbb{P}}\left(\ce(\theta,\mathbf{x}',1)-\lambda C(\mathbf{x}',\mathbf{x})\right)\right]\right\} \\
		\geq & \mathbb{E}_{\mathbf{x}\sim\mathbb{P}}\left[\max\limits_{\mathbf{x}'\in\mathcal{X}}\left(\ce(\theta,\mathbf{x}'(\mathbf{x}),1)-\lambda C(\mathbf{x}'(\mathbf{x}),\mathbf{x})\right)\right] - \zeta, 
	\end{align*}
	\end{fleqn}
	where $\zeta\geq 0$ exists as the maximum value of a distribution can have measurable distance to its expectation. As $\zeta$ is arbitrary, this gives
	\begin{fleqn}
	\begin{align*}
		&\max\limits_{\mathbb{P}':W(\mathbb{P}',\mathbb{P})\leq0}\mathbb{E}_{\mathbf{x}'\sim\mathbb{P}'}\ce(\theta,\mathbf{x}',1) \\
		=&\inf_{\lambda\geq 0}\left\{\mathbb{E}_{\mathbf{x}\sim\mathbb{P}}\left[\max\limits_{\mathbf{x}'}\left(\ce(\theta,\mathbf{x}',1)-\lambda C(\mathbf{x}',\mathbf{x})\right)\right]\right\} \\
		=&\inf_{\lambda\geq 0}\left\{\mathbb{E}_{\mathbf{x}\sim\mathbb{P}}\left[\max\limits_{\mathbf{x}'}\left(\ce(\theta,\mathbf{x}',1)-\lambda\psi_\vartheta(\mathbf{x}')+\lambda\tau\right)\right]\right\},
	\end{align*}
	\end{fleqn}
	which leads to the proposition.
\end{proof}

\begin{table}[htbp!]
\centering\caption{Summary of notations\label{table:notations}}
\begin{tabular}{l|p{.33\textwidth}}
\hline
Notation & Meaning\\\hline
$z\in\mathcal{Z}$ & software sample $z\in\mathcal{Z}$ in the space $\mathcal{Z}$ \\
$y\in\mathcal{Y}$ & ground truth label $y$ corresponding to $z$ in the space $\mathcal{Y}=\{0,1\}$\\
$\mathbf{x}\in\mathcal{X}$ & representation vector in the discrete space $\mathcal{X}$ \\
$\phi:\mathcal{Z}\to\mathcal{X}$ & feature extraction $\phi$ maps $z$ to $\mathbf{x}\in\mathcal{X}$\\
$\phi^{-1},\tilde{\phi}^{-1}$ & exact and approximate inverse feature extractions, respectively \\
$\varphi_\theta:\mathcal{X}\to\mathcal{Y}$ & ML classifier $\varphi_\theta$ maps $\mathbf{x}$ into label space $\mathcal{Y}$ \\
$f:\mathcal{Z}\to\mathcal{Y}$ & malware detector $f(\cdot)=\varphi_\theta(\phi(\cdot))$ \\
$\psi_\vartheta:\mathcal{X}\to\mathbb{R}$ & density estimator maps $\mathbf{x}$ to a real-value confidence score\\
$g:\mathcal{Z}\to\mathbb{R}$ & adversary detector $g(\cdot)=\psi_\vartheta(\phi(\cdot))$ \\
$n$& the number of dimensions of data sample $\mathbf{x}$\\
$\theta,\vartheta$ & learnable parameters of ML models \\
$\mathcal{F},\mathcal{G}$ & loss functions for $f$ and $g$, respectively \\
$\mathcal{J}:\mathcal{X}\to\mathbb{R}$ &  criterion function for attackers \\
$D_z$ & training dataset on $\mathcal{Z}\times\mathcal{Y}$, i.e., $D_z\subseteq\mathcal{Z}\times\mathcal{Y}$ \\
$D_\mathbf{x}$ & training dataset on $\mathcal{X}\times\mathcal{Y}$ corresponding to $D_z$ \\
$\delta_z, z'$ & perturbations and adversarial example in the problem space, $z'=z+\delta_z$ \\
$\delta_\mathbf{x}, \mathbf{x}',\mathbf{x}^\ast$ & perturbations and adversarial example $\mathbf{x}'=\mathbf{x}+\delta_\mathbf{x}\in\mathcal{X}$, and $\mathbf{x}^\ast$ being optimal one \\
$\mathbf{e}_p$ & a unit vector with $\|\mathbf{e}\|_p=1$ for $p$ norm \\
$\beta_1,\beta_2,\lambda$ & positive values serving as penalty factors \\
$C$ & a point-wise measurement $C:\mathcal{X}\times\mathcal{X}\to\mathbb{R}$ \\
$\mathbb{P}, \mathbb{P}'$ & the underlying distributions of $\mathbf{x}$ and $\mathbf{x}'$, respectively \\
$W$ & Wasserstein distance \\
$p=1,2,\infty$ & $\ell_p$ norm types \\
$B$ &  batch size \\
$t, T$ & $t^\text{th}$ times of $T$ iterations for attacks \\
$j, N$ & $j^\text{th}$ times of $N$ epochs for training \\ 
$\gamma$ & learning rate of optimization for training \\
${\sf L}_\mathbf{xx}^f, {\sf L}_\mathbf{x\theta}^f$ & smoothness factors of classification loss w.r.t. input \\
${\sf L}_{\theta\mathbf{x}}^f$ & smoothness factor of classification loss w.r.t. parameters \\
${\sf L}_\mathbf{xx}^g, {\sf L}_\mathbf{x\vartheta}^g$ & smoothness factors of density estimation loss w.r.t. input \\
${\sf M}_\mathbf{xx}^g$ & convexity factor of $\psi_\vartheta$ \\
\hline
\end{tabular}
\end{table}

\subsection{Theorem \ref{theorem:attack}} \label{sec:appendix:theorem1}
\begin{theorem*}
	Suppose the smoothness assumption holds. %$\mathbf{x}^\ast$ is the optimal adversarial example, and $\mathcal{J}(\mathbf{x})=\ce(\theta,\mathbf{x},\vartheta)-\lambda\psi_\vartheta(\mathbf{x})$.
	When ${\sf L}_{\bf xx}^f < \lambda {\sf M}_{\bf xx}^g$, the perturbed sample $\mathbf{x}'=\mathbf{x}+\delta_{\mathbf{x}}^{(T)}$ from Algorithm \ref{alg:adv-train} satisfies:
	$$
	\frac{\mathcal{J}(\mathbf{x}^\ast)-\mathcal{J}(\mathbf{x}')}{\mathcal{J}(\mathbf{x}^\ast)-\mathcal{J}(\mathbf{x})}\leq\exp(-\frac{T}{d}\cdot\frac{\lambda {\sf M}_{\bf xx}^g-{\sf L}_{\bf xx}^f}{\lambda {\sf L}_{\bf xx}^g+{\sf L}_{\bf xx}^f}),
	$$
	where $d$ is the dimension and $\mathcal{J}(\mathbf{x})=\ce(\theta,\mathbf{x},y)-\lambda\psi_\vartheta(\mathbf{x})$.
\end{theorem*}
	\begin{proof}
	We first present the following lemma: 
	\begin{lemma} \label{lemma:pertb}
		Given an instance-label pair $(\mathbf{x},y)$ with perturbation $\forall \delta_\mathbf{x}^{(t1)},\delta_\mathbf{x}^{(t2)} \in [\underline{\mathbf{u}}-\mathbf{x},\overline{\mathbf{u}}-\mathbf{x}]$ with $0\leq t1<t2\leq T$. We have
		\begin{align*}
		\mathcal{J}(\mathbf{x}^{(t2)})-\mathcal{J}(\mathbf{x}^{(t1)})
			\leq\frac{1/2}{\lambda {\sf M}_{\bf xx}^g-{\sf L}_{\bf xx}^f}\left\|\nabla_{\bf x}\mathcal{J}(\mathbf{x}^{(t1)})\right\|_2^2
		\end{align*}
		where $\mathbf{x}^{(t1)}=\mathbf{x}+\delta_\mathbf{x}^{(t1)}$ and $\mathbf{x}^{(t2)} = \mathbf{x}+\delta_\mathbf{x}^{(t2)}$.
		%  $\phi(z\oplus{\delta}_z^1)$ corresponds to $\mathbf{x}_1$ without considering the edge information and 
		\end{lemma}
		Based on Proposition \ref{lemma:conca-smooth}, we have
		\begin{fleqn}
		\begin{align*}
			&\mathcal{J}(\mathbf{x}^{(t2)})-\mathcal{J}(\mathbf{x}^{(t1)}) \\
			\leq &\langle\nabla_{\bf x}\mathcal{J}(\mathbf{x}^{(t1)}),\mathbf{x}^{(t2)}-\mathbf{x}^{(t1)}\rangle-\frac{\lambda {\sf M}_{\bf xx}^g - {\sf L}_{\bf xx}^f}{2}\|\mathbf{x}^{(t2)}-\mathbf{x}^{(t1)}\|_2^2 \\
			\leq & \text{\small $\max_{\mathbf{a}\in[\underline{\mathbf{u}},\overline{\mathbf{u}}]}\left(\left\langle\nabla_{\bf x}\mathcal{J}(\mathbf{x}^{(t1)}),\mathbf{a}-\mathbf{x}^{(t1)}\right\rangle-\frac{\lambda {\sf M}_{\bf xx}^g - {\sf L}_{\bf xx}^f}{2}\|\mathbf{a}-\mathbf{x}^{(t1)}\|_2^2\right)$} %\\
			% =& \frac{1}{2}\frac{1}{\lambda L_{vv}^a-L_{vv}}\left\|\left(\nabla_{v}J_\theta(\mathbf{v}_1,y)\right)_{\delta_z^2\backslash\delta_z^1}\right\|_2^2. \\
			% \leq&\frac{1}{2}\frac{1}{\lambda L_{xx}^a-L_{xx}} \left\|\nabla_{x}J_\theta(z\oplus\delta_z^1,y)\right\|_\infty^2,
		\end{align*}
		\end{fleqn}
Let $\mathbf{a}-\mathbf{x}^{(t1)}$ follow the same direction as $\nabla_{\bf x}\mathcal{J}$. We obtain the maximum $\frac{1/2}{\lambda {\sf M}_{\bf xx}^g-{\sf L}_{\bf xx}^f}\left\|\nabla_{\bf x}\mathcal{J}(\mathbf{x}^{(t1)})\right\|_2^2$ at the point $\mathbf{a}=\mathbf{x}^{(t1)}-1/({\sf L}_{\bf xx}^f-\lambda {\sf M}_{\bf xx}^g)\nabla_{\bf x}\mathcal{J}(\mathbf{x}^{(t1)})$. This leads to the lemma.
		
		\smallskip
		\noindent Further, let $p\,(p=1,2,\infty)$ norm correspond to its dual version $q\,(q=\infty,2,1)$. 
		Considering two adjacent perturbations $\delta_\mathbf{x}^{(t)}$ and $\delta_\mathbf{x}^{(t+1)}$ with $0 \leq t< T$, we can derive:
	\begin{fleqn}
	\begin{align*}
		&\mathcal{J}(\mathbf{x}+\delta_\mathbf{x}^{(t+1)}))-\mathcal{J}(\mathbf{x}+\delta_\mathbf{x}^{(t)})\\
		\geq&\langle\nabla_{\bf x}\mathcal{J}(\mathbf{x}+\delta_\mathbf{x}^{(t)})),\alpha_p\mathbf{e}_{p}\rangle-\alpha^2_p\frac{\lambda {\sf L}_{\bf xx}^g+{\sf L}_{\bf xx}^f}{2}\\
		=&\alpha_p\|\nabla_{\bf x}\mathcal{J}(\mathbf{x}+\delta_\mathbf{x}^{(t)})\|_q - \alpha^2_p\frac{\lambda {\sf L}_{\bf xx}^g+{\sf L}_{\bf xx}^f}{2} 
       \end{align*}
	\end{fleqn}
By plugging in
%Substituting 
$$\alpha_p=\frac{\|\nabla_{\bf x}\mathcal{J}(\mathbf{x}+\delta_\mathbf{x}^{(t)})\|_q}{\lambda {\sf L}_{\bf xx}^g+{\sf L}_{\bf xx}^f},$$
we have
        \begin{fleqn}
	\begin{align*}
         &\mathcal{J}(\mathbf{x}+\delta_\mathbf{x}^{(t+1)}))-\mathcal{J}(\mathbf{x}+\delta_\mathbf{x}^{(t)}) \\
	%\geq&\max_{\alpha}\left(\alpha\|\nabla_{\bf x}\mathcal{J}(\mathbf{x}+\delta_\mathbf{x}^{(t)})\|_{q}-\alpha^2\frac{{\sf L}_{\bf xx}^f+\lambda {\sf M}_{\bf xx}^g}{2}\right)\\
		\geq&\frac{1}{2\lambda {\sf L}_{\bf xx}^g+2{\sf L}_{\bf xx}^f}\left\Vert\nabla_{\bf x}\mathcal{J}(\mathbf{x}+\delta_\mathbf{x}^{(t)})\right\Vert_{q}^2 \\
		\overset{\footnotesize \whitecircled{2}}{\geq}&\frac{1}{2d\lambda {\sf L}_{\sf xx}^g+2d{\sf L}_{\bf xx}^f}\left\Vert\nabla_{\bf x}\mathcal{J}(\mathbf{x}+\delta_\mathbf{x}^{(t)})\right\Vert_2^2 \nonumber\\
		\overset{\footnotesize \whitecircled{3}}{\geq}& \frac{\lambda {\sf M}_{\bf xx}^g - {\sf L}_{\bf xx}^f}{d\lambda {\sf L}_{\bf xx}^g+d{\sf L}_{\bf xx}^f}\big(\mathcal{J}(\mathbf{x}+\delta_\mathbf{x}^\ast)-\mathcal{J}(\mathbf{x}+\delta_\mathbf{x}^{(t)})\big)
	\end{align*}
	\end{fleqn}
where {\footnotesize $\whitecircled{2}$} holds because of inequalities $\sqrt{d}\|\cdot\|_\infty\geq\|\cdot\|_2$ and $\|\cdot\|_1\geq\|\cdot\|_2$ on vector norms.
{\footnotesize $\whitecircled{3}$} holds because of Lemma \ref{lemma:pertb}, while noting that the value of $\alpha_p$ is not always held. Nevertheless, for any $\alpha_p$, we can derive certain theoretical results according to $\|\nabla_{\bf x}\mathcal{J}(\mathbf{x}+\delta_\mathbf{x}^{(t)})\|_q$, but decreasing the elegance of formulation. Furthermore, we have
	\begin{fleqn}
	\begin{align*}
	&\mathcal{J}(\mathbf{x}+\delta_\mathbf{x}^{(t+1)}) - \mathcal{J}(\mathbf{x}+\delta_\mathbf{x}^{(t)}) \\
	= &\text{\small$ \left(\mathcal{J}(\mathbf{x}+\delta_\mathbf{x}^\ast)-\mathcal{J}(\mathbf{x}+\delta_\mathbf{x}^{(t)})\right)-\left(\mathcal{J}(\mathbf{x}+\delta_\mathbf{x}^\ast)-\mathcal{J}(\mathbf{x}+\delta_\mathbf{x}^{(t+1)})\right)$}\\
	\geq &\frac{\lambda {\sf M}_{\bf xx}^g - {\sf L}_{\bf xx}^f}{d\lambda {\sf L}_{\bf xx}^g+d{\sf L}_{\bf xx}^f}\big(\mathcal{J}(\mathbf{x}+\delta_\mathbf{x}^\ast)-\mathcal{J}(\mathbf{x}+\delta_\mathbf{x}^{(t)})\big).
	\end{align*}
	\end{fleqn}
By re-organizing the preceding inequality, we obtain the gap between the optimal attack and the approximate one:
	\begin{fleqn}
	\begin{align*}
		&\mathcal{J}(\mathbf{x}^\ast)-\mathcal{J}(\mathbf{x}')=\mathcal{J}(\mathbf{x}+\delta_\mathbf{x}^\ast)-\mathcal{J}(\mathbf{x}+\delta_\mathbf{x}^{(T)}) \\
		\leq& \left(\mathcal{J}(\mathbf{x}+\delta_\mathbf{x}^\ast)-\mathcal{J}(\mathbf{x}+\delta_\mathbf{x}^{(T-1)})\right)\left(1-\frac{\lambda {\sf M}_{\bf xx}^g - {\sf L}_{\bf xx}^f}{d\lambda {\sf L}_{\bf xx}^g+d{\sf L}_{\bf xx}^f}\right)\\
		\leq& \cdots \\
		\leq&\left(\mathcal{J}(\mathbf{x}+\delta_\mathbf{x}^\ast)-\mathcal{J}(\mathbf{x}+\delta_\mathbf{x}^{(0)})\right)\left(1-\frac{\lambda {\sf M}_{\bf xx}^g - {\sf L}_{\bf xx}^f}{d\lambda {\sf L}_{\bf xx}^g+d{\sf L}_{\bf xx}^f}\right)^T\\
		\leq&(\mathcal{J}(\mathbf{x}^\ast)-\mathcal{J}(\mathbf{x}))\exp(-\frac{T}{d}\cdot\frac{\lambda {\sf M}_{\bf xx}^g - {\sf L}_{\bf xx}^f}{\lambda {\sf L}_{\bf xx}^g+{\sf L}_{\bf xx}^f}).
	\end{align*}
	\end{fleqn}
	This leads to the theorem.
\end{proof}

\subsection{Theorem \ref{theorem:convergence}} \label{sec:appendix:theorem2}

Let  $\mathcal{H}(\theta)=\mathbb{E}_{(\mathbf{x},y)\in D_\mathbf{x}}\ce(\theta,\mathbf{x}^\ast(\theta),y)$ denote the objective on the entire training dataset $D_\mathbf{x}$. Given a batch of training data samples $\{(\mathbf{x}_i,y_i)\}_{i=1}^B$, let $h(\theta)=\frac{1}{B}\sum_{i=1}^{B}\ce(\theta,\mathbf{x}^\ast_i,y_i)$ denote the mean classification loss on a batch of optimal adversarial examples. This implies that $\mathbf{x}^\ast$ is perturbed from $\mathbf{x}$ satisfying $\langle\nabla_{\bf x}\mathcal{J}(\mathbf{\bf x}^\ast), {\mathbf{x}'}-\mathbf{x}^\ast\rangle\leq0$ with ${\mathbf{x}'}$ near to $\mathbf{x}^\ast$. Indeed, the parameter $\theta$ is updated by $\theta^{(j+1)}=\theta^{(j)}-\gamma^{(j)} \nabla\hat{h}(\theta^{(j)})$, where $\hat{h}(\theta^{(j)})=\frac{1}{B}\sum_{i=1}^{B}\ce(\theta^{(j)},{\mathbf{x}}'_i)$ on perturbed examples, and $\gamma^{(j)}$ is the learning rate at $j^\text{th}$ iteration.

We additionally make an assumption of bounded gradients for SGD \cite{pmlr-v97-wang19i}. 
% Let $\mathcal{L}(\theta)=\mathbb{E}_{(\mathbf{x},y)\in D_\mathbf{x}}[\ce(\theta,\mathbf{x}^\ast,y)]$ denote the adversarial loss on the training dataset.
\begin{assumption}[Boundness assumption \cite{DBLP:conf/iclr/SinhaND18}] \label{assump:bound}
The variance of stochastic gradients is bounded by a constant $\zeta^2>0$ where
	$$\mathbb{E}(\|\nabla h(\theta)-\nabla \mathcal{H}(\theta)\|_2^2)\leq\zeta^2.$$
\end{assumption}

We first show $\mathcal{H}$ is smooth and then prove the SGD convergence under the approximate attack. Recall that ${\sf L}_{\theta \mathbf{x}}^f$ and ${\sf L}_{\theta\theta}^f$ denote the Lipschitz contant of $\nabla_\theta \ce(\theta,\mathbf{x},y)$ w.r.t $\mathbf{x}$ and $\theta$, respectively.

\begin{lemma} \label{lemma:convg}
	Let assumption \ref{assump:smooth} hold. Then, $\mathbb{E}_{(\mathbf{x},y)\in D_\mathbf{x}}\ce(\theta,\mathbf{x}^\ast,y)$ is ${\sf L}$-smooth, where ${\sf L}=\frac{{\sf L}^f_{\theta \mathbf{x}}(\lambda {\sf L}_{\mathbf{x}\theta}^g+{\sf L}_{\mathbf{x}\theta}^f)}{\lambda {\sf M}_{\bf xx}^g-{\sf L}_{\bf xx}^f}+{\sf L}^f_{\theta\theta}.$
\end{lemma}
\begin{proof}
	Given any two sets of parameters $\theta_1,\theta_2$, we have:
	\begin{fleqn}
	\begin{align}
		% &\,\|\nabla_\theta \mathcal{H}(\theta_2)-\nabla_\theta \mathcal{H}(\theta_1)\|_2 \\
		&\left\|\mathbb{E}_{(\mathbf{x},y)\in D_\mathbf{x}}\left[\nabla_\theta \ce(\theta_2, \mathbf{x}^\ast(\theta_2),y)-\nabla_\theta\ce(\theta_1,\mathbf{x}^\ast(\theta_1),y)\right]\right\|_2 \nonumber\\
		\leq &\mathbb{E}_{(\mathbf{x},y)\in D_\mathbf{x}}\left\|\nabla_\theta \ce(\theta_2, \mathbf{x}^\ast(\theta_2),y)-\nabla_\theta\ce(\theta_1,\mathbf{x}^\ast(\theta_1),y)\right\|_2 \nonumber\\
		\leq &\mathbb{E}_{(\mathbf{x},y)\in D_\mathbf{x}}\left\|\nabla_\theta \ce(\theta_2, \mathbf{x}^\ast(\theta_2),y)-\nabla_\theta \ce(\theta_2, \mathbf{x}^\ast(\theta_1),y)\right\|_2 \nonumber\\
		&+\mathbb{E}_{(\mathbf{x},y)\in D_\mathbf{x}}\left\|\nabla_\theta \ce(\theta_2, \mathbf{x}^\ast(\theta_1),y)-\nabla_\theta \ce(\theta_1, \mathbf{x}^\ast(\theta_1),y)\right\|_2 \nonumber\\
		\leq &{\sf L}^f_{\theta \mathbf{x}}\|\mathbf{x}^\ast(\theta_2)-\mathbf{x}^\ast(\theta_1)\|_2+{\sf L}^f_{\theta\theta}\|\theta_1-\theta_2\|_2. \label{eq:lemma2:1}
	\end{align}
	\end{fleqn}
	The first and second inequalities hold because of the triangle inequality. Suppose $\mathcal{J}$ is parameterized by $\theta_2$, say $\mathcal{J}_{\theta_2}$, due to its concavity, we derive
	{\small
	\begin{fleqn}
	\begin{align*}
		&\mathcal{J}_{\theta_2}(\mathbf{x}^\ast(\theta_2)) - \mathcal{J}_{\theta_2}(\mathbf{x}^\ast(\theta_1))
		\leq\big\langle\nabla_x \mathcal{J}_{\theta_2}(\mathbf{x}^\ast(\theta_1)),\mathbf{x}^\ast(\theta_2) -\mathbf{x}^\ast(\theta_1)\big\rangle\\
		&\quad\quad
		-\frac{\lambda {\sf M}_{\bf xx}^g-{\sf L}_{\bf xx}^f}{2}\|\mathbf{x}^\ast(\theta_2)-\mathbf{x}^\ast(\theta_1)\|_2^2; \\
		&\frac{\lambda {\sf M}_{\bf xx}^g-{\sf L}_{\bf xx}^f}{2}\|\mathbf{x}^\ast(\theta_2)-\mathbf{x}^\ast(\theta_1)\|_2^2\leq \mathcal{J}_{\theta_2}(\mathbf{x}^\ast(\theta_2))-\mathcal{J}_{\theta_2}(\mathbf{x}^\ast(\theta_1)).
	\end{align*}
	\end{fleqn}
	}
By combining the two inequalities, we obtain:
	\begin{fleqn}
	\begin{align}
		&(\lambda {\sf M}_{\bf xx}^g-{\sf L}_{\bf xx})\|\mathbf{x}^\ast(\theta_2)-\mathbf{x}^\ast(\theta_1)\|_2^2 \nonumber\\
		 \leq&\langle\nabla_{\bf x}\mathcal{J}_{\theta_2}(\mathbf{x}^\ast(\theta_1)),\mathbf{x}^\ast(\theta_2) -\mathbf{x}^\ast(\theta_1)\rangle \nonumber\\
		\overset{\footnotesize \whitecircled{4}}{\leq}& \left\langle\nabla_{\bf x}\mathcal{J}_{\theta_2}(\mathbf{x}^\ast(\theta_1))-\nabla_{\bf x}\mathcal{J}_{\theta_1}(\mathbf{x}^\ast(\theta_1)),\mathbf{x}^\ast(\theta_2) -\mathbf{x}^\ast(\theta_1)\right\rangle \nonumber\\
		\overset{\footnotesize \whitecircled{5}}{\leq}&\left\|\nabla_{\bf x}\mathcal{J}_{\theta_2}(\mathbf{x}^\ast(\theta_1))-\nabla_{\bf x}\mathcal{J}_{\theta_1}(\mathbf{x}^\ast(\theta_1))\right\|_2\left\|\mathbf{x}^\ast(\theta_2) -\mathbf{x}^\ast(\theta_1)\right\|_2 \nonumber\\
		\overset{\footnotesize \whitecircled{6}}{\leq}& ({\sf L}_{\mathbf{x}\theta}^f+\lambda {\sf L}_{\mathbf{x}\theta}^g)\left\|\theta_1-\theta_2\right\|_2\left\|\mathbf{x}^\ast(\theta_2) -\mathbf{x}^\ast(\theta_1)\right\|_2 \label{eq:lemm2:2} 
	\end{align}
	\end{fleqn}
{\footnotesize where \whitecircled{4}} holds as $\langle\nabla_{\bf x}J_{\theta_1}(\mathbf{x}^\ast(\theta_1)),\mathbf{x}^\ast(\theta_2) -\mathbf{x}^\ast(\theta_1)\rangle\leq0$, {\footnotesize \whitecircled{5}} holds because of the Cauchy-Schwarz inequality, and {\footnotesize \whitecircled{6}} holds as $\mathcal{J}_{\theta_2}$ is $({\sf L}_{\mathbf{x}\theta}^f+\lambda {\sf L}_{\mathbf{x}\theta}^g)$-smooth. Combining Eq.\eqref{eq:lemma2:1} and Eq.\eqref{eq:lemm2:2} leads to
	\begin{align*}
		\frac{\|\nabla\mathcal{H}(\theta_1)-\nabla \mathcal{H}(\theta_2)\|_2}{\|\theta_1-\theta_2\|_2}\leq\bigg(\frac{{\sf L}^f_{\theta \mathbf{x}}(\lambda {\sf L}_{\mathbf{x}\theta}^g+{\sf L}_{\mathbf{x}\theta}^f)}{\lambda {\sf M}_{\bf xx}^g-{\sf L}_{\bf xx}^f}
		+{\sf L}^f_{\theta\theta}\bigg).
	\end{align*} 
\end{proof}

% Upon Lemma \ref{lemma:convg}, we present the following theorem, which demonstrates the convergence rate of adversarial training is sublinear \cite{DBLP:conf/iclr/SinhaND18}.
% Moreover, the non-optimal attacks have effects on the robustness dependent on their difference of attack effectiveness between the optimal attack (see Theorem \ref{theorem:attack}).
\begin{theorem*} 
Let $\Delta=\mathcal{H}(\theta^{(0)})-\min_\theta \mathcal{H}(\theta)$. Under Assumption \ref{assump:smooth} and Assumption \ref{assump:bound}, if we set the learning rate to $\gamma^{(j)}=\gamma=\minimum(1/{\sf L}, \sqrt{\Delta/({\sf L}\zeta^2N)}$, the adversarial training satisfies
	\begin{equation}
		\frac{1}{N}\sum_{j=0}^{N}\mathbb{E}\left\|\nabla \mathcal{H}(\theta^{(j)})\right\|\leq \zeta\sqrt{8\frac{\Delta {\sf L}}{N}} + 2\hat{c},
	\end{equation} 
	where $N$ is the epochs (i.e., the total iterations of SGD), and 
 %$\hat{c}$ 
 $\hat{c}=(\mathcal{J}(\mathbf{x}^\ast)-\mathcal{J}({\mathbf{x}}))\frac{ 2{\sf L}^f_{\theta \mathbf{x}}}{\lambda {\sf M}_{\bf xx}^g-{\sf L}_{\bf xx}^f}\exp(\frac{T}{d}\frac{{\sf L}_{\bf xx}^f-\lambda {\sf M}_{\bf xx}^g}{\lambda {\sf L}_{\bf xx}^g+{\sf L}_{\bf xx}^f})$ is a constant. 
	% and $\minimum$ is element-wise operation for returning the minimum element between two elements.
\end{theorem*}
\begin{proof}
Inspired \cite{DBLP:conf/iclr/SinhaND18}, we derive the following at the $j^\text{th}$ iteration:
\begin{fleqn}
\begin{align*}
	&\mathcal{H}(\theta^{(j+1)})\\
	\leq&\mathcal{H}(\theta^{(j)})+\langle\nabla \mathcal{H}(\theta^{(j)}),\theta^{(j+1)}-\theta^{(j)}\rangle+\frac{\sf L}{2}\|\theta^{(j+1)}-\theta^{(j)}\|_2^2 \\
	=&\mathcal{H}(\theta^{(j)})+\gamma\langle\nabla\mathcal{H}(\theta^{(j)}),\nabla\mathcal{H}(\theta^{(j)})-\nabla\hat{h}(\theta^{(j)})\rangle \\
	&\,-\gamma\|\nabla\mathcal{H}(\theta^{(j)})\|_2^2+\frac{{\sf L}\gamma^2}{2}\|\nabla\hat{h}(\theta^{(j)})\|_2^2 \\
	=&\mathcal{H}(\theta^{(j)})+(\gamma-{\sf L}\gamma^2)\left\langle\nabla\mathcal{H}(\theta^{(j)}),\nabla\mathcal{H}(\theta^{(j)})-\nabla\hat{h}(\theta^{(j)})\right\rangle \\
	&\,+\frac{{\sf L}\gamma^2}{2}\|\nabla\hat{h}(\theta^{(j)})-\nabla\mathcal{H}(\theta^{(j)})\|_2^2 \\
	&\,-(\gamma-\frac{{\sf L}\gamma^2}{2})\|\nabla\mathcal{H}(\theta^{(j)})\|_2^2\\
	=& \mathcal{H}(\theta^{(j)})+(\gamma-{\sf L}\gamma^2)\left\langle\nabla\mathcal{H}(\theta^{(j)}),\nabla\mathcal{H}(\theta^{(j)})-\nabla{h}(\theta^{(j)})\right\rangle\\
	&\,+{\sf L}\gamma^2\|\nabla h(\theta^{(j)})-\nabla\mathcal{H}(\theta^{(j)})\|_2^2-\frac{\gamma}{2}\|\nabla\mathcal{H}(\theta^{(j)})\|_2^2 \\
	&\,+\frac{\gamma+{\sf L}\gamma^2}{2}\|\nabla h(\theta^{(j)})-\nabla\hat{h}(\theta^{(j)})\|_2^2.
	% \leq& \mathcal{L}^\ast(\theta^t)+(\frac{\gamma}{2}-L^\ast\gamma^2)\|\nabla\mathcal{L}^\ast(\theta^t)\|_2^2+\frac{\gamma}{2}\|b(\theta^t)-\nabla\mathcal{L}^\ast(\theta^t)\|_2^2 \\
	% &\quad\quad +(\gamma-\frac{L^\ast\gamma^2}{2})\|\nabla b(\theta^t)-\nabla\hat{b}(\theta^t)\|_2^2
\end{align*}
\end{fleqn}
Taking conditional expectations of $\mathcal{H}(\theta^{(j+1)})-\mathcal{H}(\theta^{(j)})$ on $\theta^{(j)}$ and using $\mathbb{E}(\nabla h(\theta^{(j)}))=\mathcal{H}(\theta^{(j)})$, we have 
\begin{eqnarray}
	&\mathbb{E}(\mathcal{H}(\theta^{(j+1)})-\mathcal{H}(\theta^{(j)})|\theta^{(j)})\leq -\frac{\gamma}{2}\mathbb{E}(\|\nabla\mathcal{H}(\theta^{(j)})\|_2^2) \nonumber\\
	&+{\sf L}\gamma^2\zeta^2+\frac{\gamma+{\sf L}\gamma^2}{2}\|\nabla h(\theta^{(j)})-\nabla\hat{h}(\theta^{(j)})\|_2^2. \label{eq:converge-tele}
\end{eqnarray}
Furthermore, we derive
\begin{align*}
	&\|\nabla h(\theta^{(j)})-\nabla\hat{h}(\theta^{(j)})\|_2^2 \\ =&\left\|\frac{1}{B}\sum_{i=1}^{B}\nabla_\theta \ce(\theta^{(j)},{\mathbf{x}'}_i,y_i) - \frac{1}{B}\sum_{i=1}^{B}\nabla_\theta \ce(\theta^{(j)},\mathbf{x}^\ast_i,y_i)\right\|_2^2 \\
	\leq& \frac{1}{B}\sum_{i=1}^{B}\left\|\nabla_\theta \ce(\theta^{(j)},{\mathbf{x}'}_i,y_i)-\nabla_\theta\ce(\theta^{(j)},{\mathbf{x}}^\ast_i,y_i)\right\|_2^2 \\
	\leq& \frac{1}{B}\sum_{i=1}^{B}{\sf L}^f_{\theta \mathbf{x}}\|\mathbf{x}'_i-\mathbf{x}^\ast_i\|_2^2 \\
	\leq& \frac{1}{B}\sum_{i=1}^{B}\frac{2{\sf L}^f_{\theta \mathbf{x}}}{\lambda {\sf M}_{\bf xx}^g-{\sf L}_{\bf xx}}\left(\mathcal{J}(\mathbf{x}^\ast_i)-\mathcal{J}(\mathbf{x}'_i)\right) \\
	\leq& (\mathcal{J}(\mathbf{x}^\ast_i)-\mathcal{J}(\mathbf{x}_i))\frac{ 2{\sf L}^f_{\theta \mathbf{x}}}{\lambda {\sf M}_{\bf xx}^g-{\sf L}_{\bf xx}^f}\exp(\frac{T}{d}\frac{{\sf L}_{\bf xx}^f-\lambda {\sf M}_{\bf xx}^g}{{\sf L}_{\bf xx}^f+\lambda {\sf L}_{\bf xx}^g}) = \hat{c}.
\end{align*}
Plugging the preceding inequities into Ineq.\eqref{eq:converge-tele} and taking telescope sum of it over $j=0,\ldots,N-1$, we obtain
\begin{align*}
	\frac{1}{N}\sum_{j=0}^{N-1}\mathbb{E}(\|\nabla\mathcal{H}(\theta^{(j)})\|_2^2)\leq&\frac{2}{\gamma N}\mathbb{E}(\mathcal{H}(\theta^{(0)})-\mathcal{H}(\theta^{(N)}))\\
	&+2{\sf L}\gamma\zeta^2
	+(1+{\sf L}\gamma)\hat{c}. 
\end{align*} 
Using the fact $\gamma \leq \frac{1}{\sf L}$ and $\mathcal{H}(\theta^{(0)})-\mathcal{H}(\theta^{(N)})\leq\mathcal{H}(\theta^{(0)})-\min_\theta \mathcal{H}(\theta)=\Delta$, we have
\begin{align*}
	\frac{1}{N}\sum_{i=0}^{N-1}\mathbb{E}(\|\nabla\mathcal{H}(\theta^{(i)})\|_2^2)&\leq\frac{2\Delta}{\gamma N}+2{\sf L}\gamma\zeta^2
	+2\hat{c} \\
	&\leq \min_{\gamma} \left(\frac{2\Delta}{\gamma N}+2{\sf L}\gamma\zeta^2
	+2\hat{c}\right) \\
	&= \zeta\sqrt{8\frac{\Delta {\sf L}}{N}} + 2\hat{c},
\end{align*}
where $\gamma=\sqrt{\frac{\Delta}{{\sf L}\zeta^2}N}$.
This leads to the theorem.
\end{proof}

\section{Additional Experimental Analysis}
\subsection{Manipulation Example} \label{append:me}

We show how to manipulate malware examples by conducting perturbations in the feature space. For manifest features, we inject them into the file AndroidManifest.xml by following the defined format. For API features, we leverage an example to illustrate the manipulation. Listing \ref{lst:java-example} shows the malware gets the device ID and then sends sensitive information from the phone to the outside world via SMS. We observe that apps (e.g., the one with md5 checksum {\tt 4cc8****f212} and the one with {\tt f07d****3b7b}) use this pattern to retrieve a user's private information. In order to mislead malware detectors, Listing \ref{lst:api-ist} shows how to inject irrelevant APIs into the code snippet, and Listing \ref{lst:api-rmv} hides {\tt sendTextMessage} using {\tt Java} reflection, both of which retain the malicious functionality.

\makeatletter
\setlength{\@fptop}{0pt}
\makeatother
\begin{figure}[!t]
	\begin{lstlisting}[
		numbers=none,
		xleftmargin=.11\columnwidth,
		xrightmargin=.11\columnwidth,
		caption=Sending sensitive information via SMS,
		label={lst:java-example}, 
		captionpos=b
		]
		TelephonyManager telecom = // default ;
		String str = "";
		if (Build.VERSION.SDK_INT >= Build.VERSION_CODES.O) {
			str = telephonyMgr.getImei();
		} else {
			str = telecom.getDeviceId();
		}
		SmsManager smgr = SmsManager.getDefault();
		smgr.sendTextMessage("97605", null, str, null, null);
	\end{lstlisting}
	% \vspace{-0.8em}
	\begin{lstlisting}[
		numbers=none,
		xleftmargin=.11\columnwidth,
		xrightmargin=.11\columnwidth,
		caption=API insertion,
		label={lst:api-ist}, 
		captionpos=b
		]
		if (False){
			try {
				ConnectivityManager cmgr = null;
				NetworkInfo anet = cmgr.getActiveNetworkInfo();
			} catch (Exception e) {}
		}
	}
	\end{lstlisting}
	%\vspace{-1em}
	\begin{lstlisting}[
		numbers=none,
		xleftmargin=.11\columnwidth,
		xrightmargin=.11\columnwidth,
		caption=API removal,
		label={lst:api-rmv}, 
		captionpos=b
		]
		String mtd_name = "sendTextMessage";
		Method send_sms = null;
		send_sms = smgr.getClass().getMethod(mtd_name, String.class, String.class, String.class, PendingIntent.class, PendingIntent.class);
		send_sms.invoke(smgr, "97605", null, str, null, null);
	\end{lstlisting}
	
	\begin{tikzpicture}[overlay,remember picture]
		\draw[-latex,dashed,line width=1,black!60] (1.8,9.25) -- (0.6, 9.25) -- (0.6,6.3) -- (0.9,6.3);
		\node(add)[label=below:{},align=center] at (0.38, 7.2) {\circled{1}};
		\draw[-latex,dashed,line width=1,black!60] (7.3,9.1) -- (8.3, 9.1) -- (8.3,2.6) -- (8, 2.6);
		\draw[dashed,line width=1,black!60] (7.3,8.8) edge (8.3, 8.8);
		\node(rmv)[label=below:{},align=center] at (8.62, 3.9) {\circled{2}};
	\end{tikzpicture}
	
	\vspace{-10pt}
	\caption{Code snippets for perturbing apps. Manipulation \protect\scalebox{0.6}{\protect\circled{1}} inserts junk codes before sending text messages and manipulation \protect\scalebox{0.6}{\protect\circled{2}} hides the {\tt sendTextMessage} using {\tt Java} reflection.}
	\label{fig:code-example}
	% \vspace{-10pt}
\end{figure}

\subsection{Training Time and Test Time}

We implement the defense models using PyTroch libraries \cite{DBLP:conf/nips/PaszkeGMLBCKLGA19} and run experiments on a CUDA-enabled GTX 2080 Ti GPU and Intel Xeon W-2133 CPU@3.60GHz.

Figure \ref{fig:train-time} reports the training time of the defenses. We observe that adversarial training-based defenses take much longer than standard training without involving adversarial examples. This is because searching for perturbations is conducted per iteration in standard training. Furthermore, AT-MaxMA and PAD-SMA leverage several attacks to produce adversarial examples and thus require more time. Since PAD-SMA encapsulates not only a malware detector but also an adversary detector, the longest cost is consumed.

Furthermore, we report the Mean Test Time to Detection (MTTD) for PAD-SMA. We ignore the other defenses because all models share the same feature extraction method and the ML part runs very fast. Using the Drebin test dataset, MTTD of PAD-SMA is 1.72s using 1 CPU core and 0.52s using 6 CPU cores. Using the Malscan test dataset, MTTD of PAD-SMA is 8.91s using 1 CPU core and 2.79s using 6 CPU cores. Our model may not hit the limit of the user's patience, particularly when multi-core computing is available, because the test time within 5s is reasonable \cite{10.1145/3433667.3433669}.

\begin{figure}[!t]
	\centering
	\begin{subfigure}[b]{0.48\textwidth}
		\centering
	\includegraphics[width=\textwidth]{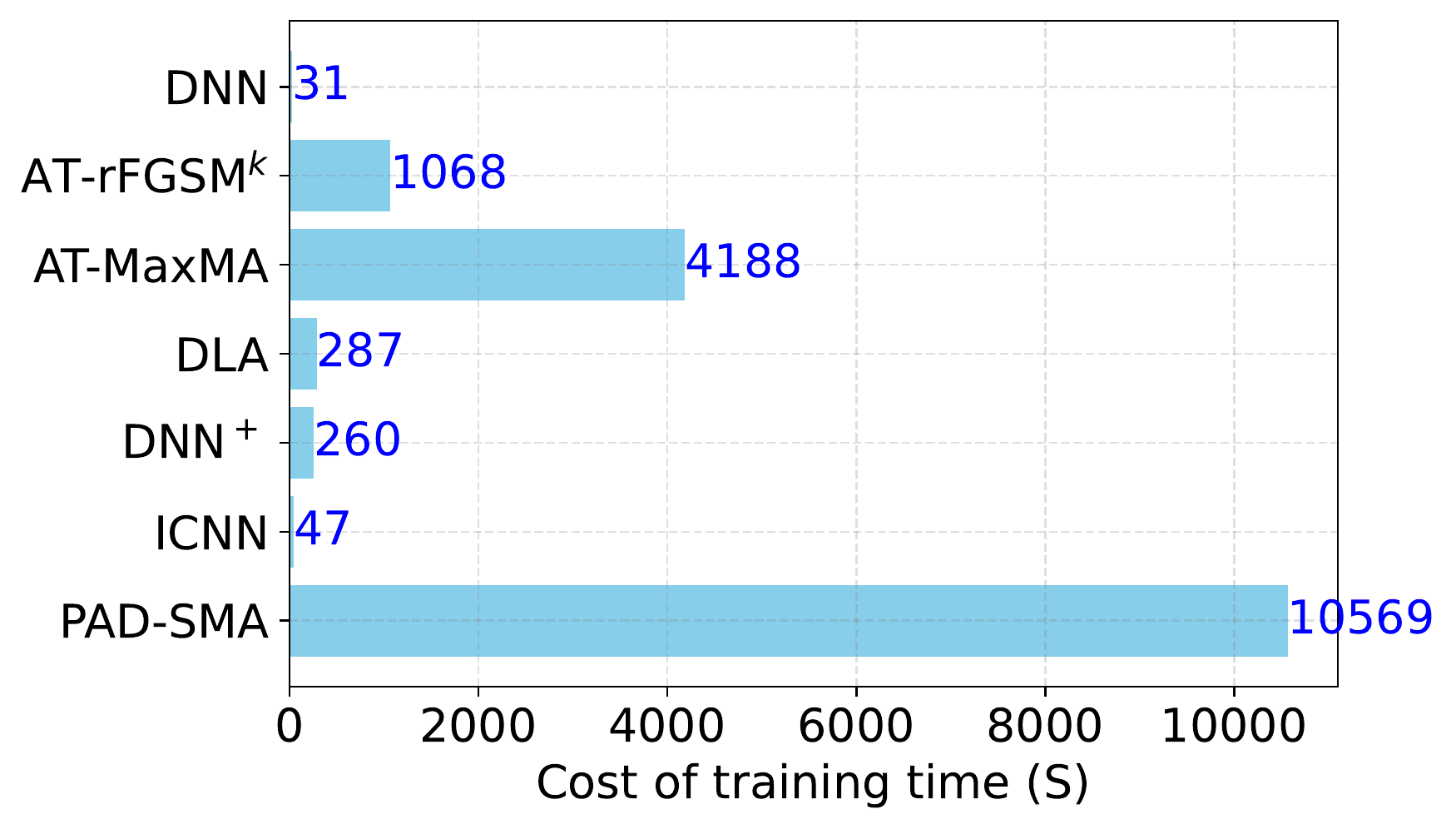}
		\caption{Drebin}
		% \label{}
	\end{subfigure} \vspace{-8pt}	 
	\begin{subfigure}[b]{0.48\textwidth}
		\centering
	   % \vspace{-3pt}	 
        \includegraphics[width=\textwidth]{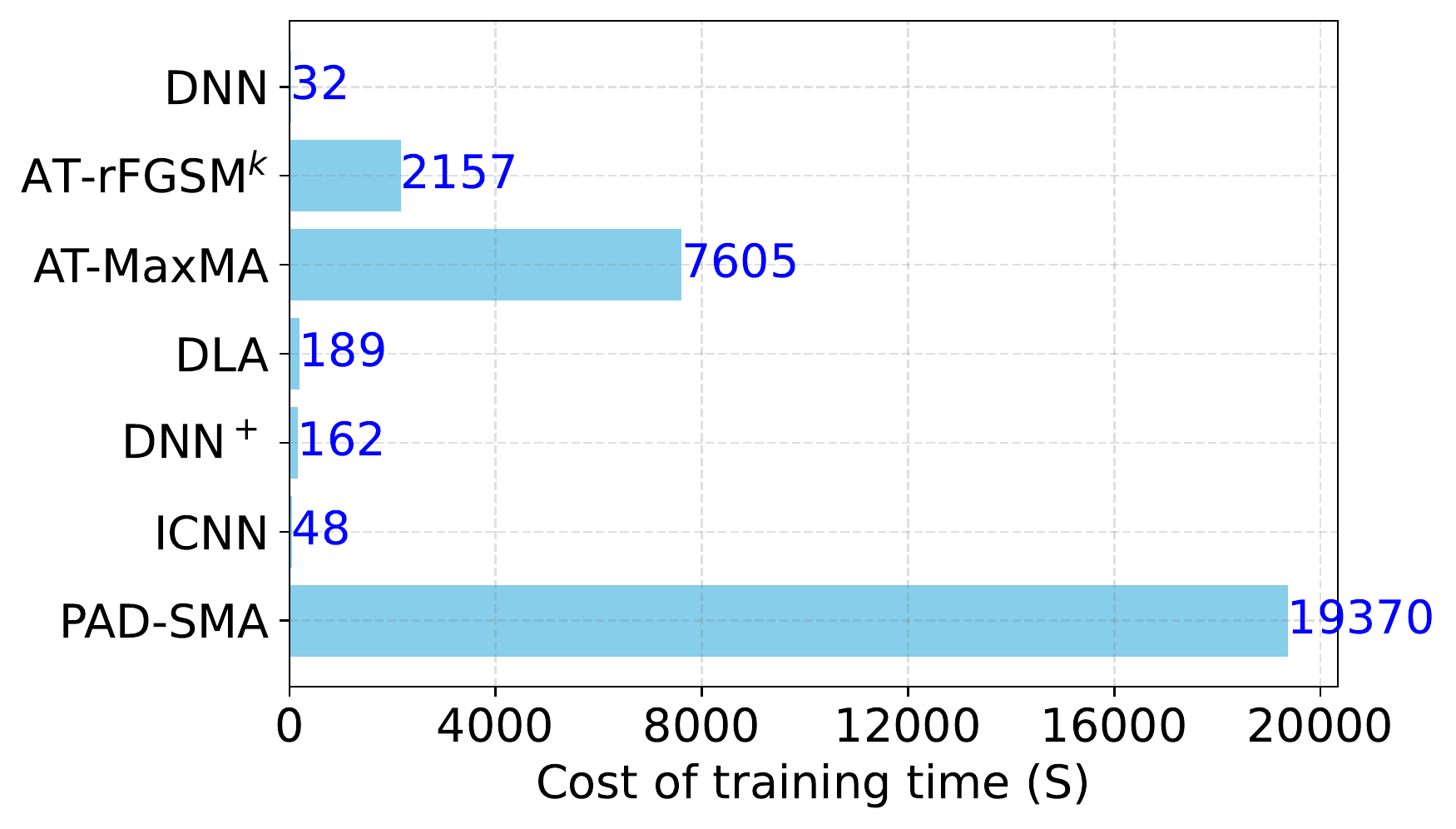}
		\caption{Malscan}
		%\label{subfig-attack:02}
	\end{subfigure}
	\caption{The cost of training time for defenses.
     }
    \label{fig:train-time}
\end{figure}

\end{document}